\documentclass[submission,copyright,creativecommons]{eptcs}
\providecommand{\event}{GandALF 2017} % Name of the event you are submitting to
\usepackage{underscore}           % Only needed if you use pdflatex.
\usepackage{microtype}
\usepackage{xspace}
\usepackage{tikz}
\usetikzlibrary{shapes,decorations,arrows,automata,petri,positioning,calc}
\usepackage{paralist}
\usepackage{colortbl}
\usepackage{arydshln}
\usepackage{anyfontsize}
\usepackage{stmaryrd}
\usepackage{amssymb}
\usepackage{amsthm}
\usepackage{amsmath}
\usepackage{graphicx}
\usepackage{wrapfig}

\SetSymbolFont{stmry}{bold}{U}{stmry}{m}{n}
\newcommand{\cut}[1]{}
\newcommand{\cutdavid}[1]{}

\newcommand{\cA}{\mathcal{A}}

\newcommand{\cC}{\mathcal{C}}

\newcommand{\cL}{\mathcal{L}}

\newcommand{\cN}{\mathcal{N}}

\newcommand{\cP}{\mathcal{P}}

\newcommand{\cS}{\mathcal{S}}

\newcommand{\bbN}{\mathbb{N}}
\newcommand{\bbNp}{\mathbb{N}_{>0}}

\newcommand{\omegals}{$\omega$-regular languages\xspace}

\newcommand{\omegats}{$\omega T$-regular languages\xspace}

\newcommand{\pnl}{{\rm PNL}}     % PNL

\newcommand{\oL}{\overline{L}}

\newcommand{\os}{\overline{s}}

\newcommand{\seq}[1]{\ensuremath{\vec{#1}}\xspace}

\newcommand{\semantics}[1]{{\ensuremath{\llbracket #1 \rrbracket}}\xspace}
\newcommand{\semanticsvarphi}{\semantics{\varphi}}

\newcommand{\Uquantifier}{{\ensuremath{\mathbb U}}\xspace}
\newcommand{\Bquantifier}{{\ensuremath{\mathbb B}}\xspace}

\newcommand{\msoconstant}{\textsf{MSO}}
\newcommand{\mso}{\msoconstant\xspace}

\newcommand{\sonesconstant}{\textsf{S1S}}
\newcommand{\wsonesconstant}{\textit{w}\textsf{S1S}}
\newcommand{\sones}{\sonesconstant\xspace}
\newcommand{\wsones}{\wsonesconstant\xspace}
\newcommand{\sonesU}{\sonesconstant+\textsf{U}\xspace}
\newcommand{\wsonesU}{\wsonesconstant+\textsf{U}\xspace}

\newcommand{\infcheck}[2]{{\ensuremath{\mathit{check}^{\infty}_{#1,#2}}}\xspace}
\newcommand{\infcheckCk}{\infcheck{\cC}{k}}
\newcommand{\infcheckpik}{\infcheck{\pi}{k}}

\newcommand{\cca}{\textrm{CCA}\xspace}
\newcommand{\ccas}{\textrm{CCA}s\xspace}
\newcommand{\cqa}{\textrm{CQA}\xspace}

\newcommand{\cvectorsymbol}[1][v]{\ensuremath{{\bf #1}}}
\newcommand{\cvectorsymbolsup}[2][v]{\ensuremath{\cvectorsymbol[#1]^{#2}}}
\newcommand{\cvectorsymbolsub}[2][v]{\ensuremath{\cvectorsymbol[#1]_{#2}}}

\newcommand{\cvector}[1][v]{{\cvectorsymbol[#1]}\xspace}
\newcommand{\cvectorsup}[2][v]{{\cvectorsymbolsup[#1]{#2}}\xspace}
\newcommand{\cvectorsub}[2][v]{{\cvectorsymbolsub[#1]{#2}}\xspace}

\newcommand{\cvectorv}{\cvector[v]}
\newcommand{\cvectorvsup}[1]{\cvectorsup[v]{#1}}
\newcommand{\cvectorvsub}[1]{\cvectorsub[v]{#1}}

\newcommand{\cvectorinitial}{\cvectorvsub{0}}
\newcommand{\cvectorvprime}{\cvectorvsup{\prime}}

\newcommand{\cvectorvalueconstruction}[2]{{\ensuremath{#1 {[#2]}}}\xspace}
\newcommand{\cvectorinitialvalue}[1]
{\cvectorvalueconstruction{\cvectorinitial}{#1}}
\newcommand{\cvectorvalue}[2][v]{\cvectorvalueconstruction{\cvector[#1]}{#2}}

\newcommand{\cvectorvaluesub}[3][v]
{\cvectorvalueconstruction{\cvectorsub[#1]{#2}}{#3}}

\newcommand{\cvectorvaluevi}{\cvectorvalue[v]{i}}

\newcommand{\cvectorvaluevsub}[2]{\cvectorvaluesub{#1}{#2}}

\newcommand{\cvectorvaluevsubik}{\cvectorvaluesub[v]{i}{k}}

\newcommand{\Prefixes}[1]{{\ensuremath{\mathit{Prefixes}_{#1}}}\xspace}
\newcommand{\PrefixesA}{\Prefixes{\mathcal{A}}}

\newcommand{\indexbegin}{\ensuremath{\mathit{begin}}\xspace}

\newcommand{\indexend}{\ensuremath{\mathit{end}}\xspace}

\newcommand{\languagewitness}[1]{{\ensuremath{\cL_{\mathsf{w}}(#1)}}\xspace}
\newcommand{\languagewitnessA}{\languagewitness{\cA}}

\newcommand{\languagenfa}[1]{{\ensuremath{\cL(#1)}}\xspace}
\newcommand{\languagenfaN}{\languagenfa{\cN}}

\newcommand{\Si}{\Sigma}
\newcommand{\De}{\Delta}
\newcommand{\de}{\delta}

\newcommand{\isregexp}[1]
{{\ensuremath{\mathit{is\_reg\_exp_{#1}}}}\xspace}

\newcommand{\isbeginningofexp}[1]
{{\ensuremath{\mathit{Beginning\_of_{#1}}}}\xspace}

\newcommand{\Tcondition}[1]{{\ensuremath{\Phi_{#1}^{T\mathit{cond}}}}\xspace}
\newcommand{\block}[1]{{\ensuremath{\Phi_{#1\text{-}\mathit{block}}}}\xspace}
\newcommand{\blockset}[1]{{\ensuremath{\Phi_{#1\text{-}\mathit{block}\text{-}\mathit{set}}}}\xspace}
\newcommand{\ismaxblock}[2]{{\ensuremath{is\_maximal\_#1block(#2)}}\xspace}
\newcommand{\ismaxblockE}[1]{\ismaxblock{E}{#1}}
\newcommand{\ismaxblockEX}{\ismaxblockE{X}}

\newcommand{\noneps}{\ensuremath{\mathit{non}\text{-}\epsilon}}

\newtheorem{proposition}{Proposition}

\newtheorem{definition}{Definition}
\newtheorem{lemma}{Lemma}
\newtheorem{theorem}{Theorem}

\title{Beyond {$\omega{BS}$}-regular Languages: 
  {$\omega{T}$}-regular \\ Expressions and Counter-Check Automata\footnote{This work was partially supported by the Italian INdAM-GNCS project {\em Logics and Automata for Interval Model Checking}. In addition, D. Della Monica acknowledges the financial support from a Marie Curie INdAM-COFUND-2012 Outgoing Fellowship.}}

\author{Dario Della Monica
\institute{Universidad Complutense de Madrid, Spain, and}
\institute{Universit\`a ``Federico II'' di Napoli, Italy.}
\email{{ddellamo@ucm.es}}
\and
Angelo Montanari
\institute{Universit\`a di Udine, Italy.}
\email{{angelo.montanari@uniud.it}}
\and
Pietro Sala
\institute{Universit\`a di Verona, Italy.}
\email{{pietro.sala@univr.it}}
}
\def\titlerunning{Beyond $\omega{BS}$-regular Languages: $\omega{T}$-regular Expressions and Counter-Check Automata}
\def\authorrunning{D. Della Monica, A. Montanari, and P. Sala}
\begin{document}
\maketitle

\begin{abstract}
  In the last years, various extensions of $\omega$-regular languages have been proposed in the literature, 
  including $\omega{B}$-regular ($\omega$-regular languages extended with boundedness), $\omega{S}$-regular 
  ($\omega$-regular languages extended with strict unboundedness), and $\omega{BS}$-regular languages (the 
  combination of $\omega{B}$- and $\omega{S}$-regular ones). 
  While the first two classes satisfy a generalized closure property, namely, the complement 
  of an $\omega{B}$-regular (resp., $\omega{S}$-regular) language is an $\omega{S}$-regular 
  (resp., $\omega{B}$-regular) one, the last class is not closed under complementation.
  The existence of non-$\omega{BS}$-regular languages that are the complements
  of some $\omega{BS}$-regular ones and express fairly natural properties of
  reactive systems motivates the search for other well-behaved classes of extended 
  $\omega$-regular languages.
  In this paper, we introduce the class of $\omega{T}$-regular languages, that
  includes meaningful languages which are not $\omega{BS}$-regular.
  We first define it in terms of $\omega{T}$-regular expressions. Then, we introduce a new class 
  of automata (counter-check automata) and we prove that (i) their emptiness 
  problem is decidable in PTIME and (ii) they are expressive enough to capture $\omega{T}$-regular languages
  (whether or not $\omega{T}$-regular languages are expressively complete with
  respect to counter-check automata is still an open problem).
  Finally, we provide an encoding of $\omega{T}$-regular expressions into \sonesU.
\end{abstract}

\section{Introduction}\label{sec:intro}
A fundamental role
in  computer science is played by $\omega$-regular languages, as they provide a natural setting for the specification and verification of nonterminating finite-state systems. Since the seminal work by B\"uchi~\cite{Buchi62}, McNaughton~\cite{McNaughton66}, and Elgot and Rabin \cite{DBLP:journals/jsyml/ElgotR66} in the sixties, a great research effort has been devoted to the theory and the applications of $\omega$-regular languages. Equivalent characterisations of $\omega$-regular languages have been given in terms of formal languages (\emph{$\omega$-regular  expressions}), automata (B\"uchi, Rabin, and Muller automata), classical logic (weak/strong monadic second-order logic of one successor, \wsones/\sones for short), and temporal logic (Quantified Linear Temporal Logic, Extended Temporal Logic).

Recently, it has been shown that $\omega$-regular languages can be extended in various ways, preserving their decidability and some of their closure properties \cite{bojan11tcs,DBLP:conf/lics/BojanczykC06}.
As an example, extended $\omega$-regular languages make it possible to constrain the distance between consecutive occurrences of a given symbol to be (un)bounded (in the limit). Boundedness comes into play in the study of \emph{finitary fairness} as opposed to the classic notion of \emph{fairness}, widely used in automated verification of concurrent systems. According to the latter, no individual process in a multi-process system may be ignored for ever; finitary fairness imposes the stronger constraint that every enabled transition is executed within at most $b$ time-units, where $b$ is an unknown, constant bound. In~\cite{DBLP:journals/toplas/AlurH98}, it is shown that finitary fairness enjoys some desirable mathematical properties that are violated by the weaker notion of fairness, and yet it captures all reasonable schedulers' implementations. The same property has been investigated from a logical perspective in~\cite{DBLP:journals/fmsd/KupfermanPV09}, where the logic PROMPT-LTL is introduced. Roughly speaking, PROMPT-LTL extends LTL with the \emph{prompt-eventually} operator, which states that an event will happen within the next $b$ time-units, $b$ being an unknown, constant bound.
An analogous extension has been recently proposed for the propositional interval logic of temporal neighborhood \pnl~\cite{jelia16}.

From the point of view of formal languages, the proposed extensions pair the Kleene star $(.)^*$ with bounding/unbounding variants of it. Intuitively, the bounding exponent $(.)^B$ (aka $B$-constructor) constrains parts
of the input word to be of bounded size, while the unbounding exponent $(.)^S$ (aka $S$-constructor) forces parts
of the input word to be arbitrarily large. The two extensions have been studied both in isolation ($\omega{B}$- and $\omega{S}$-regular expressions) and in conjunction ($\omega{BS}$-regular expressions). Equivalent characterisations of extended $\omega$-regular languages are given in~\cite{bojan11tcs,DBLP:conf/lics/BojanczykC06} in terms of automata ($\omega{B}$-, $\omega{S}$-, and $\omega{BS}$-automata) and classical logic (fragments of \wsonesU, i.e., the extension of \wsones with the unbounding quantifier \Uquantifier \cite{bojan04csl}, that allows one to express properties which are satisfied by finite sets of arbitrarily large size).\footnote{Undecidability of full \sonesU has been shown in \cite{DBLP:conf/stacs/BojanczykPT16}.}
In~\cite{DBLP:conf/lics/BojanczykC06}, the authors also show that the complement of an $\omega{B}$-regular language is an $\omega{S}$-regular one and vice versa; moreover, they show that $\omega{BS}$-regular languages, featuring both $B$- and $S$-constructors, strictly extend $\omega{B}$- and $\omega{S}$-regular languages and are not closed under complementation. 

%\smallskip

In this paper, we focus on those $\omega$-languages which are complements of $\omega BS$-regular ones, but are not $\omega BS$-regular. 
%Our ultimate goal is to provide a characterisation of the class of these languages, and the present work is a nontrivial step in this direction. 
We start with an in-depth analysis of a paradigmatic example of one such language~\cite{DBLP:conf/lics/BojanczykC06}.
%the complement of an $\omega{BS}$-regular language that lies outside the class of $\omega{BS}$-regular languages~\cite{DBLP:conf/lics/BojanczykC06}.
It allows us to identify a meaningful extension of $\omega$-regular languages ($\omega{T}$-regular languages) including it and obtained by adding a new, fairly natural constructor  $(.)^T$, named $T$-constructor,
to the standard constructors of $\omega$-regular expressions. 
%We conjecture that the class obtained by enriching $\omega{BS}$-regular languages with the $T$-constructor is closed under complementationSuch a belief that the $T$-constructor actually fills the gap left by $\omega{BS}$-regular languages is supported by the fact that--FALSE 
An interesting feature of such a class is that pairing $(.)^B$ and $(.)^S$ with $(.)^T$ one can capture all possible ways of instantiating  $*$-expressions (this is not the case with $B$ and $S$ only). In view of that, it can be said that $(.)^T$ ``complements'' $(.)^B$ and $(.)^S$ with respect to $(.)^*$.
% it covers the missing behaviors of the Kleene star
Then, we introduce a new class of automata (counter-check automata), that are expressive enough to capture $\omega{T}$-regular languages, and we show that their emptiness problem is decidable.
Finally, we provide an encoding of $\omega{T}$-regular expressions (languages) into \sonesU.

The paper is organized as follows. In Section~\ref{sec:languages}, we illustrate existing extensions of $\omega$-regular languages, with a special attention to $\omega{BS}$-regular ones, and we introduce the class 
of $\omega{T}$-regular languages. In Section~\ref{sec:automata}, we define counter-check automata (\cca) and prove that their emptiness problem is decidable in PTIME. In Section~\ref{sec:encoding}, we provide an encoding of  \texorpdfstring{$\omega T$}{omega T}-regular languages into \cca, while, in Section~\ref{sec:logic}, we show that they can be defined in \sonesU. 
%In Section \ref{sec:strongT}, we briefly discuss a stronger variant of $(.)^T$. 
Conclusions provide an assessment of the work done and outline future research directions.

%%% Local Variables: 
%%% TeX-master: "main"
%%% End: 

\section{Extensions of
\texorpdfstring{\omegals}{omega-regular languages}}
\label{sec:languages}

In this section, we 
%first 
give a short account of the extensions of \texorpdfstring{\omegals}{omega-regular languages}
proposed in the literature (details can be found in~\cite{bojan11tcs,DBLP:conf/lics/BojanczykC06})
and 
%then 
we outline a new 
%meaningful 
one.
%
%At a very intuitive level,
To begin with, we observe that an 
%infinite word (
\emph{$\omega$-word}
%)
can be seen as the concatenation of a finite prefix,
belonging to a regular language,
%followed by
and an infinite sequence of finite words (we call each of these finite words an \emph{$\omega$-iteration}),
also belonging to a regular language.
%
%Both the prefix language and the language of the $\omega$-iterations
%%of the infinite suffix
%can be given as regular expressions  ($\omega$-regular expression).
A standard way to define $\omega$-regular languages is by means of $\omega$-regular expressions.
An interesting case is that of $\omega$-iterations consisting of a finite sequence of words,
%(\emph{$*$-iterations}),
generated by an occurrence of the Kleene
star operator $(.)^*$, aka \emph{$*$-constructor}, in the scope of the \emph{$\omega$-constructor}
$(.)^\omega$.
%Since $\omega$-iterations are expressed by regular expressions, each of them
%could in turn feature a finite sequence of words (\emph{$*$-iterations}),
%which are iterated by using the Kleene star operator $(.)^*$, aka
%\emph{$*$-constructor}, in the scope of the \emph{$\omega$-constructor}
%$(.)^\omega$.
As an example, the $\omega$-regular expression $(a^*b)^\omega$ generates the
language of $\omega$-words featuring an infinite sequence of $\omega$-iterations,
each one consisting of a finite (possibly empty) sequence of $a$'s followed by exactly
one $b$.
%
%Informally speaking, given
Given an $\omega$-regular expression $E$ featuring an occurrence of $(.)^*$ (sub-expression $R^*$) in the scope of  $(.)^\omega$
%containing $R^*$ as a sub-expression (i.e., $E$ features an occurrence of the
%$*$-constructor in the scope of the $\omega$-constructor),
and an $\omega$-word $w$ belonging to the language of $E$,
%described by $E$,
we refer to the sequence of the sizes of the (maximal) blocks of consecutive iterations
of $R$ in the different $\omega$-iterations as the \emph{(sequence of) exponents of $R$
in (the $\omega$-iterations of) $w$}. Let $w = abaabaaab \ldots$ be an $\omega$-word 
%$w = abaabaaabaaaab \ldots$,
%captured by the above
generated by the 
%above 
$\omega$-regular expression $(a^*b)^\omega$.
The sequence of exponents of $a$ in $w$ is
%$1, 2, 3, 4, \ldots$. 
$1, 2, 3, \ldots$. Sometimes, we will denote words in a compact way, by
explicitly indicating the exponents of a sub-expression, e.g., we will write
$w$ as $a^1ba^2ba^3b \ldots$.
%$a^1ba^2ba^3ba^4b \ldots$.

Given an expression $E$, we denote by $\mathcal L(E)$ the language defined by $E$. With a little abuse of notation, we will sometimes identify a language with the expression defining it, and vice versa, e.g., 
%so, for instance, 
we will 
%simply 
write ``language $(a^* b)^\omega$'' instead of ``language $\mathcal L((a^* b)^\omega)$''.
Notice that $(.)^*$ 
%operator 
allows one to impose the existence of a finite sequence of words (described by its argument
expression) within each $\omega$-iteration, but it cannot be used to express properties 
%on 
of the sequence of exponents of its argument expression in the $\omega$-iterations of an $\omega$-word.
%
%$\omega$-regular expressions define the class of languages known as
%\omegals. Intuitively, these are languages of infinite words, obtained through
%the infinite concatenation of finite words expressed by regular expressions.
%
To overcome such a limitation, some meaningful extensions of
$\omega$-regular expressions have been investigated in the last years, that make
it possible to constrain the behavior of $(.)^*$  in the limit.
%
% to deal with particular iterations of the Kleene star $(.)^*$
%
%In this section, we provide a short account of the proposed extensions (details
%can be found in \cite{bojan04csl,bojan11tcs,DBLP:conf/lics/BojanczykC06})
%and we outline a new one.

\smallskip

\noindent \textbf{Beyond \texorpdfstring{$\omega$}{omega}-regularity.}
\label{subsec:beyond}
A first class of extended $\omega$-regular languages is that
of $\omega{B}$-regular languages, that allow one to impose boundedness
conditions.~$\omega{B}$-regular expressions are obtained from $\omega$-regular
ones by adding a variant of $(.)^*$, called \emph{$B$-constructor}
and denoted by $(.)^B$, to be used in the scope of $(.)^\omega$.
The bounded exponent $B$ allows one to constrain the argument $R$ of the
expression $R^B$ to be repeated in each $\omega$-iteration a number of times
less than a certain bound fixed for the whole $\omega$-word.
As an example, % \cite{DBLP:conf/lics/BojanczykC06},
the expression $(a^B b)^\omega$ denotes the language of $\omega$-words
in $(a^\ast b)^\omega$ for  which there is an upper bound on the number of
consecutive occurrences of $a$ (the sequence of exponents of $a$ is bounded).
As the bound may vary from word to word, the language is not $\omega$-regular.
The class of  $\omega{S}$-regular languages extends that of $\omega$-regular
ones with strong unboundedness. By analogy with $\omega{B}$-regular
expressions, $\omega{S}$-regular expressions are obtained from $\omega$-regular
ones by adding a variant of $(.)^*$, called \emph{$S$-constructor}
and denoted by $(.)^S$, to be used in the scope of $(.)^\omega$.
For every $\omega{S}$-regular expression containing the sub-expression $R^S$
and every natural number $k > 0$, the strictly
%strongly
%\marginpar{\tiny ``strictly'' instead of ``strongly''? to be uniform
%with~\cite{DBLP:conf/lics/BojanczykC06}???}
unbounded exponent $S$ constrains
the number of $\omega$-iterations in which the argument $R$ is repeated at most $k$ times to be finite.
%
%constrains every exponent of the argument $R$  to appear only a finite number of times.
%
Let us consider $\omega$-words that feature an infinite
number of instantiations of the expression $R^S$, that is, $\omega$-words
%\marginpar{\scriptsize
%We should clarify why this is not always the case. At first glance, it seems
%that this is always the case and thus this restriction does not make so much
%sense}
for which there exists an infinite number of $\omega$-iterations including a sequence
of consecutive $R$'s generated by $R^S$. It can be easily checked that in these
words the sequence of exponents of $R$ tends towards infinity.
%, for every
%infinite chain of (maximal) sequences of consecutive $R$'s generated by $R^S$,
%the infimum of the size of such sequences tends toward infinity.
%
%$\omega$-words which features infinite application of the $S$
%operator must force every maximal chains of exponents of its  argument $R$
%to have an unbounded infimum.
%
As an example, % \cite{DBLP:conf/lics/BojanczykC06},
the expression $(a^S b)^\omega$ denotes the language of $\omega$-words
$w$ in $(a^\ast b)^\omega$ such that, for any 
%natural number 
$k > 0$, there exists
a suffix of $w$ that only features maximal sequences of consecutive $a$'s that
are longer than $k$.

$\omega{BS}$-regular expressions are built by using the operators of
$\omega$-regular expressions and both $(.)^B$ and $(.)^S$.
%Consequently,
In~\cite{DBLP:conf/lics/BojanczykC06},
%Boja\'nczyk and Colcombet
the authors show that the class of $\omega BS$-regular languages
%these are more expressive than their component
%languages (i.e., the $\omega B$-regular and the $\omega S$-regular languages).
strictly includes the classes of $\omega{B}$- and $\omega{S}$-regular languages
as witnessed by the $\omega{BS}$-regular language $L = (a^B b + a^S b )^\omega$
consisting of those $\omega$-words $w$ featuring infinitely many occurrences of
$b$ and such that there are only finitely many numbers occurring infinitely
often in the sequence of exponents of $a$ in $w$, that is, there is a
bound $k$ such that no $h > k$ occurs infinitely often in the sequence of exponents of $a$ in
$w$. $L$ is neither $\omega B$- nor $\omega S$-regular.\footnote{The constructor $+$ occurring in $L$ must not be thought of as performing the union of two languages, but rather as a ``shuffling operator'' that mixes $\omega$-iterations belonging to the two different (sub-)languages.}
%This will be made clear later on, when we will formally define the languages we deal with.}
Moreover, they prove that the class of $\omega{BS}$-regular languages is not
closed under complementation.
A counterexample is given precisely by 
%the $\omega{BS}$-regular language 
$L$, whose complement is
not $\omega BS$-regular (notice that $\omega{BS}$-regular languages whose
complement is not an $\omega{BS}$-regular language are neither $\omega B$-
nor $\omega S$-regular languages, as the complement of an $\omega{B}$-regular
language is an $\omega{S}$-regular one and vice versa).
% in~\cite{DBLP:conf/lics/BojanczykC06},

\label{sec:languageLdefinition}

% $\omega{BS}$-regular
%expressions are obtained from $\omega$-regular ones by adding both $(.)^B$
%and $(.)^S$ exponents to be used in the scope of the $\omega$-constructor
%$(.)^\omega$.

%\smallskip

In this paper, we investigate those $\omega$-languages that do not belong to the class
of $\omega{BS}$-regular languages, but whose complement belongs to this class.
%To have some insights into these languages, 
Let us consider, for instance, the complement $\oL$ of the language $L$ above.
%On the one hand, 
%any  $\omega$-word $w$ in $\oL$ that features an infinite
%number of occurrences of $b$ must feature an infinite sequence of blocks of consecutive $a$'s  (between two consecutive $b$'s) of unbounded size; otherwise,  $w$ would belong to $L$, as it would be captured by the sub-expression $a^B b$. On the other hand, for any such $\omega$-word $w$, there must be a natural number
%$k > 0$ such that there exist infinitely many maximal blocks of consecutive $a$'s whose size is exactly $k$;
%otherwise, $w$ would belong to $L$, as it would be captured by the sub-expression $a^S b$.
%Thus, $w$ is such that
%\begin{inparaenum}[$(i)$]
%\item for every natural number $k$, there exists $k' > k$ that occurs in the
%  sequence of exponents of $a$ in $w$, and
%\item there exists at least one natural number $k > 0$ that occurs infinitely
%  often in the sequence of exponents of $a$ in $w$.
%\end{inparaenum}
%In fact, as an effect of the combined use of both $B$- and $S$-constructors,
%$w$ is subject to an even stronger constraint:
Any word $w$ in $\oL$ that features infinitely many occurrences of $b$ (i.e.,
$w \in (a^*b)^\omega$) is such
that there are infinitely many
natural numbers that occur infinitely often in the sequence of exponents of
$a$ in $w$.
% (notice that this latter constraint implies both the former ones).
%
By way of contradiction, suppose that there are only finitely many.
% natural numbers (exponents) that occur infinitely often. 
Let $k$ be the largest one.
Now, $w$ can be viewed as an infinite sequence of
$\omega$-iterations, each of them characterised by the corresponding exponent
of $a$.
If the exponent associated with an $\omega$-iteration is greater than $k$, then
it does not occur infinitely often, and thus the $\omega$-iteration is captured by
the sub-expression $a^Sb$. Otherwise, if the exponent is not greater than $k$, then the
corresponding $\omega$-iteration is captured by the sub-expression $a^Bb$.
As an example, the $\omega$-word $a^1ba^2ba^1ba^3ba^1ba^4b \ldots$ does not belong to
$\oL$ as 1 is the only exponent occurring infinitely often, while the $\omega$-word $a^1ba^2ba^1ba^2ba^3ba^1ba^2ba^3ba^4b \ldots$ \emph{does} belong to it as 
infinitely many (actually all) natural numbers occur infinitely often in the sequence 
of exponents.
%
%All these subsequences cannot feature an exponent for the $a$'s a  finite number
%of times otherwise they will be  all captured by the language $a^S b$. Then there exists a subsequence $s$ which features at least an exponent an infinite number of times together with the fact that such the exponents of such a subsequence are not bounded.
%How many of different exponents we need to repeat infinitely many times in $s$?
%If we take a finite number of them it turns out that we can
%take the maximum $m$ of such infinitely repeated exponents split the subsequence $s$ into two  infinite ones $s'$ and $s''$. $s'$ contains all the exponent less or equal $m$ and $s''$ all the exponents greater or equal than $m$
%and thus $s'$ is captured from $a^B b$ and $s''$ from $a^S b$. Then $s$ must feature an infinite number of exponents which are repeated infinitely many times.

%\smallskip

Here, we focus on $\omega$-words featuring infinitely many exponents occurring infinitely
often.
%In order to characterise the class of languages of $\omega$-words
%with infinitely many exponents occurring infinitely often,
More precisely, we introduce a new variant of $(.)^*$, called \emph{$T$-constructor} and denoted by $(.)^T$,
to be used in the scope of $(.)^\omega$, and we define the corresponding class of extended $\omega$-regular languages (\omegats).
%$\omega{T}$-regular expressions are obtained from $\omega$-regular ones by introducing a variant of Kleene star $(.)^*$, called \emph{$T$-constructor} and denoted by $(.)^T$, to be used in the scope of the $\omega$-constructor $(.)^\omega$.
Let $E$ be an $\omega$-expression and let $w \in E$. An expression $R^T$ occurring in $E$ forces
%two conditions on the $\omega$-words $w$ belonging to $E$:
%\begin{inparaenum}[$(i)$]
%\item the sequence of exponents of $R$ in $w$ features an infinite number of
%  distinct exponents, and
%%the set of exponents occurring in the sequence is infinite.
%\item every exponent occurring in the sequence occurs infinitely often.
%\end{inparaenum}
\label{sec:intuitivedefinitionofTconstructor}
the sequence of exponents in $w$ to feature infinitely many different elements occurring infinitely often.
As an example, it can be easily checked that the language $\overline L$ can be defined as
%$((a^*b)^*a^T b)^\omega + (a^*b^*)^* a^\omega$,
$(a^T b)^\omega + (a^*b^*)^* a^\omega$,
and thus it belongs to the class of  \omegats.
%
%The  exponent $T$ allows one to constrain the argument $R$ of the expression $R^T$ to be repeated an unbounded number of times in any $\omega$-iteration and there is no bounds on the number of such exponents, in other words once we have expressed an exponent $k$ in an $\omega$-teration the $k$ iteration of $R$ is constrained to be repeated infinitely many times in the rest of the word plus the fact that there exist infinitely many of such exponents. It is easy to see that the language $\cL$
%of the previous example is captured by $((a^*b)^*a^T b)^\omega$ (DA SOSTITUIRE CON $(a^T b)^\omega$)
%which guarantees that there is at most one infinite subsequence of exponents in the word which features an infinite number of exponents and each exponent is
%repeated infinitely many times.
%
%There exists an infinite number of natural numbers $k$ such that there
%are infinitely many  $\omega$-iterations in which the argument $R$ of the
%expression $R^T$ is repeated exactly $k$ times. Notice that there may
%exist other natural numbers (possibly infinitely many) such that there
%exists a finite number of $\omega$-iterations (possibly $0$) in which the
%argument $R$ of the expression $R^T$ is repeated exactly $k$ time.
%
In the following, we first provide a formal account of
$\omega BS$-regular languages~\cite{DBLP:conf/lics/BojanczykC06} and
then we define $\omega T$-regular ones.

%We are ready now to formally define the languages that are the focus of this
%paper. In what follows, we first define $\omega BS$-regular
%expressions~\cite{DBLP:conf/lics/BojanczykC06}, and then we introduce
%the construct $(.)^T$.

\smallskip

\noindent \textbf{\texorpdfstring{$\omega BS$}{omega BS}-regular languages.}
%Intuitively, $\omega BS$-regular expressions are defined as infinite iterations
%of $BS$-regular ones.
The class of \emph{$\omega BS$-regular languages} is the class of languages
defined by $\omega BS$-regular expressions.
These latter are built on top of $BS$-regular expressions, just as
$\omega$-regular expressions are built on top of regular ones.
Let $\Sigma$ be a finite, non-empty alphabet.
A \emph{$BS$-regular expression} over $\Sigma$ is defined by the
grammar~\cite{DBLP:conf/lics/BojanczykC06}:

{\centering $ e \ ::= \ \emptyset \ \mid \ a \ \mid \ e \cdot e \ \mid \ e+e \ \mid \ e^* \ \mid
\ e^B \ \mid \ e^S $

}

\noindent with $a \in \Sigma$. Sometimes we omit the
concatenation operator, thus writing $ee$ for $e \cdot e$.

%Syntactically, 
$BS$-regular expressions differ from standard regular ones for the
presence of the 
%two 
constructors $(.)^B$ and $(.)^S$.
Since these operators constrain the behavior of the sequence of $\omega$-iterations to the
limit, it is not possible to simply define the semantics of $BS$-regular expressions in
terms of languages of (finite) words, and then to obtain $\omega BS$-regular languages through
infinitely many, unrelated iterations of such words.
Instead, we specify their semantics 
%of $BS$-regular expressions 
in terms of languages of infinite sequences of finite words; suitable constraints are imposed to such sequences 
in order to capture the intended meaning of $(.)^B$ and $(.)^S$
% the $B$- and $S$-constructors.

Let $\mathbb N$ be the set of natural numbers, including $0$, and
$\bbNp = \mathbb N \setminus \{ 0 \}$.
For an infinite sequence $\seq{u}$ of finite words over $\Sigma$, we denote by
$u_i$ ($i \in \bbNp$) its $i$-th element.
%Let $f : \mathbb N \rightarrow \mathbb N$
%with $f(0) = 1$.
%
The semantics of $BS$-regular expressions over $\Sigma$ is defined as follows
(hereafter we assume $f(0)=1$):
\begin{compactitem}
\item $\mathcal L(\emptyset) = \emptyset$;

%\item $\mathcal L(a) = \{ (a,a,a, \ldots) \}$, that is, the infinite sequence
%of the one-letter word $a$, for each $a \in \Sigma$;
\item for $a \in \Sigma$, $\mathcal L(a)$ only contains the infinite sequence of the one-letter word $a$
 $\{(a,a,a, \ldots)\}$;

\item $\mathcal L(e_1 \cdot e_2) = \{ \seq{w} \mid \forall i. w_i = u_i \cdot v_i,
  \ \seq{u} \in \mathcal L(e_1), \ \seq{v} \in \mathcal L(e_2) \}$;

\item $\mathcal L(e_1 + e_2) = \{ \seq{w} \mid \forall i.w_i \in
  \{ u_i, v_i \},
  \ \seq{u}, \seq{v} \in \mathcal L(e_1) \cup \mathcal L(e_2) \}$;\footnote{Unlike the case of word languages,
  when applied to languages of word sequences, the operator $+$ does not return the union of the two
  argument languages. As an example, $\mathcal L(a) \cup \mathcal L(b) \subsetneq \mathcal L(a+b)$,
  as witnessed by the word sequence $(a,b,a,b,a,b,\ldots)$.
  In general, for all $BS$-regular expressions $e_1,e_2$, it holds that $\mathcal L(e_1) \cup \mathcal L(e_2) \subseteq \mathcal
  L(e_1+e_2)$.}
%
%Notice that this definition slightly differs from the one
%provided  in~\cite{DBLP:conf/lics/BojanczykC06}.
%However, the two definitions are equivalent, thanks
%to~\cite[Fact 2.1]{DBLP:conf/lics/BojanczykC06}.

\item $\mathcal L(e^*) = \{ (u_{f(0)}u_2 \ldots u_{f(1)-1}, u_{f(1)} \ldots u_{f(2)-1},
  \ldots) \mid \seq{u} \in \mathcal L(e)$ and
  $f: \mathbb N \rightarrow \bbNp$
  is an unbounded and nondecreasing function$\}$;
%\item $\mathcal L(e^*) = \{ (u_{f(0)}u_2 \ldots u_{f(1)-1}, u_{f(1)} \ldots u_{f(2)-1},
%  \ldots) \mid \\
%  \text{\hspace{25mm}} \vec{u} \in \mathcal L(e), \ f \text{ unbounded and nondecreasing} \}$;

\item $\mathcal L(e^B) = \{ (u_{f(0)}u_2 \ldots u_{f(1)-1},
  u_{f(1)} \ldots u_{f(2)-1}, \ldots) \mid \seq{u} \in \mathcal L(e)$ and
  $f: \mathbb N \rightarrow \bbNp$
  is an unbounded and nondecreasing function such that
  $\exists n \in \mathbb N \ \forall i \in \bbN . (f(i+1) - f(i) < n) \}$;
%\item $\mathcal L(e^B) = \{ (u_{f(0)}u_2 \ldots u_{f(1)-1},
%  u_{f(1)} \ldots u_{f(2)-1}, \ldots) \mid \\
%  \text{\hspace{25mm}} \vec{u} \in \mathcal L(e), \ f
%  \text{ unbounded, nondecreasing, and } \\
%  \text{\hspace{25mm}such that } \exists n \in \mathbb N \ \forall i . (f(i+1) - f(i) < n) \}$;

\item $\mathcal L(e^S) = \{ (u_{f(0)}u_2 \ldots u_{f(1)-1}, u_{f(1)} \ldots u_{f(2)-1},
  \ldots) \mid \seq{u} \in \mathcal L(e)$ and
  $f: \mathbb N \rightarrow \bbNp$
  is an unbounded and nondecreasing function such that
  $\forall n \in \mathbb N \ \exists k \in \bbN \ \forall i>k .
  (f(i+1) - f(i) > n) \}$.
%\item $\mathcal L(e^S) = \{ (u_{f(0)}u_2 \ldots u_{f(1)-1}, u_{f(1)} \ldots u_{f(2)-1},
%  \ldots) \mid \\
%  \text{\hspace{25mm}} \vec{u} \in \mathcal L(e), \ f
%  \text{ unbounded, nondecreasing, and } \\
%   \text{\hspace{25mm}such that } \forall n \in \mathbb N \ \exists k \ \forall i>k .
%  (f(i+1) - f(i) > n) \}$.

\end{compactitem}

Given a sequence $\seq v = (u_{f(0)}u_2 \ldots u_{f(1)-1}, u_{f(1)} \ldots
\allowbreak u_{f(2)-1}, \ldots) \in e^{\mathit{op}}$, where $\mathit{op} \in
 \{ *, B, S \}$, we
formally define the \emph{sequence of exponents of $e$ in $\seq{v}$}, denoted by
$N(\seq v)$, as the sequence $\Big( f(i+1) - f(i) \Big)_{i \in \bbN}$.
%
%It is worth emphasising that, while
While the $*$-constructor does not impose any
%property
constraint on the sequence of exponents of its operand,
%the $B$- and the $S$-constructors do. More precisely,
the $B$-constructor forces the sequence of exponents to be bounded
and the $S$-constructor forces it to be strictly unbounded, that is, its limit
inferior tends towards infinity (equivalently, the $S$-constructor imposes that  no
exponent occurs infinitely many times in the sequence).

The \emph{$\omega$-constructor} defines languages of infinite words from
languages of infinite word sequences. Let $e$ be a $BS$-regular expression.
The semantics of the $\omega$-constructor is defined as follows:
\begin{compactitem}
\item $\mathcal L(e^\omega) = \{ w \mid w = u_1u_2u_3 \ldots
  \text{ for some } \seq u \in \mathcal L(e) \}$.
\end{compactitem}

%Finally,
\emph{$\omega BS$-expressions} are defined by the grammar (we
%adopt the typographical convention of denoting
denote languages of word sequences by lowercase letters, such as $e$, $e_1$, \ldots,
and languages of words by uppercase ones,  such as $E$, $E_1$, \ldots, $R$, $R_1$, \ldots):

{\centering
$ E \ ::= \ E+E \ \mid \ R \cdot E \ \mid \ e^\omega $

}

\noindent where $R$ is a
%standard
regular expression,
%(see~\cite{} for a definition) and
$e$ is a $BS$-regular expression, and  
%the operators 
$+$ and $\cdot$ respectively denote union and concatenation of word languages
(formally, $\mathcal L(E_1 + E_2) = \mathcal L(E_1) \cup \mathcal L(E_2)$
and $\mathcal L(E_1 \cdot E_2) = \{ u \cdot v \mid u \in \mathcal L(E_1), v \in \mathcal L(E_2) \}$).\footnote{Notice the abuse of notation with the previous definition of the operators $+$ and $\cdot$ over languages of word sequences.}
As we did in the case of languages of word sequences, we will sometimes omit the concatenation operator between word languages.
%of words.

\smallskip

\noindent \textbf{\texorpdfstring{$\omega T$}{omega T}-regular languages.} \label{sec:omegatl}
We are now ready to introduce \texorpdfstring{$\omega T$}{omega T}-regular languages. From~\cite{DBLP:conf/lics/BojanczykC06}, we know that the class of $\omega BS$-regular languages is not closed under complementation, that is, there are $\omega$-languages that are the complements of $\omega BS$-regular ones while being not $\omega BS$-regular. This is the case, for instance, with the complement $\oL$ of the $\omega BS$-regular language $L = (a^Bb + a^Sb)^\omega$.
%(see Subsection~\ref{subsec:beyond}).
%In Subsection~\ref{sec:languageLdefinition}, we 
We have already pointed out the distinctive features of $\oL$, showing that $\omega$-words belonging to it are, to a certain extent, characterised by sequences of exponents where infinitely many exponents occur infinitely often.
In order to capture extended $\omega$-regular languages that satisfy such a property, we define a new class of \omegals, called \emph{\omegats}. It includes those languages that can be expressed by \emph{$\omega T$-regular expressions}, which are defined by the grammar (where $R$ is a regular expression and $a \in \Sigma$):

\smallskip

{\centering \phantomsection\label{page:Tregexp-grammar}
$
\begin{array}{lll}
  E & ::= & E + E \ \mid \ R \cdot E \ \mid \ e^\omega \\
  e & ::= & \emptyset \ \mid \ a \ \mid \ e \cdot e \ \mid \ e+e \ \mid
            \ e^* \ \mid \ e^T
\end{array}
$

}

\smallskip

%\noindent where $R$ is a regular expression and $a \in \Sigma$.

The sub-grammar rooted in the non-terminal $e$ generates the \emph{$T$-regular expressions}. The only new ingredient in the above definition is the $T$-constructor $(.)^T$, that, given a language of word sequences $e$, defines the following language:
\begin{compactitem}
\item $\mathcal L(e^T) = \{ (u_{f(0)}u_2 \ldots u_{f(1)-1}, u_{f(1)} \ldots u_{f(2)-1},
  \ldots) \mid \seq{u} \in \mathcal L(e)$ and
  $f: \mathbb N \rightarrow \bbNp$
  is an unbounded and nondecreasing function such that
%  \\ $
%%  \text{\hspace{5mm}}\vec{u} \in \mathcal L(t), \ f
%%  \text{ unbounded, nondecreasing, and such that:} \\
%%
%  \begin{array}{l@{\hspace{1mm}}l}
%    (i) & \forall n \exists i . f(i+1) - f(i) > n \text{, and} \\
%    (ii) & \forall n . [ \text{if } \exists i. f(i+1) - f(i) = n, \\
%         &  \text{then } \forall k \exists j > k . f(j+1) - f(j) = n ] \}.
%  \end{array}
%%
%  $
%
  $\exists^\omega n \in \bbN \ \forall k \in \bbN
  \ \exists i>k . (f(i+1) - f(i) = n) \},$
\end{compactitem}
where $\exists^\omega$ is a shorthand for ``there are infinitely many''.

For $\seq u \in e^T$, we define the \emph{sequence of exponents of $e$ in
$\seq{u}$}, denoted by $N(\seq u)$, exactly as we did in the case of 
%in the exact same way we have done above in the context of 
$BS$-regular expressions.
Moreover, for $\mathit{op} \in \{ *, B, S, T \}$ and $\seq u \in e^{\mathit{op}}$,
 we denote by
$N_{i}(\seq u)$ (resp., $N_{f}(\seq u)$) the set of exponents
occurring infinitely (resp., finitely) many times in $N(\seq u)$.
It is not difficult to see that the cardinality of $N_i(\seq u)$ is infinite, for every $\seq u \in e^T$, and thus the formal semantics of the $T$-constructor conforms with the intuitive one  given at the end of Subsection~\ref{sec:intuitivedefinitionofTconstructor}.
%Indeed,
%item $(i)$ guarantees the existence of infinitely many distinct exponents in the
%sequence and item $(ii)$ forces each exponent (occurring at least once) to
%occur infinitely many times in the sequence.

%For an example of an $\omega T$-regular language that is not $\omega BS$-regular and, vice versa, of an $\omega BS$-regular language that is not $\omega T$-regular, please refer to Appendix \ref{app:example}.
It is not difficult to devise an $\omega T$-regular language that is not $\omega BS$-regular and, vice versa, of an $\omega BS$-regular language that is not $\omega T$-regular.

%\newcounter{examplecounter}
%\setcounter{examplecounter}{\thetheorem}
%\begin{example} \label{ex:languages}
%An example of an $\omega T$-regular language, whose complement is not $\omega BS$-regular, is the language $L_1 = (a^Tb)^\omega$ (we provide evidences supporting this claim at the end of Section~\ref{sec:encoding}).
%
%Next,  let $L_2$ be the $\omega BS$-regular language  $(a^*b)^*a^\omega$ $+ ((a^*b)^* (a^Bb))^\omega$.
%It includes
%\begin{inparaenum}[$(i)$]
%\item words featuring only finitely many occurrences of $b$ (sub-expression $(a^*b)^*a^\omega$) and
%\item words featuring infinitely many occurrences of the sub-string $a^kb$, for at least one natural number $k \geq 0$ (and thus infinitely many occurrences of $b$, too).
%\end{inparaenum}
%Its complement $\overline{L_2}$ is the language of words featuring infinitely many occurrences of $b$ and only finitely many occurrences of the sub-string $a^kb$, for all $k \in \mathbb N$, that is, $\overline{L_2}$ is the language $(a^Sb)^\omega$, which is not $\omega T$-regular (we provide evidences supporting this claim in Appendix~\ref{app:example}).
%
%\end{example}
%\noindent We will provide supporting evidences for the status of $L_1$ and $L_2$ and their complements at the end of Section \ref{sec:encoding}.

%The (at least theoretical) interest of the $T$-constructor stems from the fact that it somehow 
%complements the $B$- and the $S$-constructor in several ways.
As we already pointed out in the introduction, one of the motivations for the proposal of the $T$-constructor stems from the fact that it somehow complements the other two with respect to the Kleene star. We can make such a claim more precise as follows. Let $\seq u \in \mathcal L(e^{\mathit{op}})$, with $\mathit{op} \in \{ B, S, T \}$. If $\seq u \in \mathcal L(e^{B})$, then $N(\seq u)$ is bounded, while if either $\seq u \in \mathcal L(e^{S})$ or $\seq u \in \mathcal L(e^{T})$ it is unbounded; moreover, if $\seq u \in \mathcal L(e^{S})$, then $N_i(\seq u) = \emptyset$, while if $\seq u \in \mathcal L(e^{T})$, then
%$N_f(\vec u) = \emptyset$.
$N_i(\seq u)$ is infinite.
%Indeed, due to the way its semantics is defined, it can be viewed as the counterpart of the $S$-constructor 
%with respect to word sequences with an unbounded sequence of exponents:
The next proposition shows that when paired with $(.)^B$ and $(.)^S$,  $(.)^T$ makes it possible to define the Kleene star
%
%  (proof in Appendix~\ref{app:proof-proposition-Kleene}).
%
 Let \emph{$BST$-regular expressions} be obtained from $BS$-regular ones by enriching them with $(.)^T$.
%
%\newcommand{\propstardefinition}{
%For every $BST$-regular expression $e$, it holds that $e^* = e^B + e^S + e^T$.
%}
%\newcounter{propstardefinitioncounter}
%\setcounter{propstardefinitioncounter}{\thetheorem}
%\begin{proposition} \label{prop:star-definition}
%\propstardefinition
%\end{proposition}
%\setcounter{theorem}{\thepropstardefinitioncounter}
%\begin{proposition}
%\propstardefinition
%\end{proposition}

\newcommand{\propositionKleene}{
  For every $BST$-regular expression $e$, it holds that $e^* = e^B + e^S + e^T$.
}
\newcounter{propositionKleenecounter}
\setcounter{propositionKleenecounter}{\thetheorem}
\begin{proposition}\label{prop:star-definition}
\propositionKleene
\end{proposition}
\begin{proof}
As $\mathcal L(e^B) \subseteq \mathcal L(e^*)$, $\mathcal L(e^S) \subseteq \mathcal L(e^*)$, and $\mathcal L(e^T) \subseteq \mathcal L(e^*)$, it trivially holds that $\mathcal L(e^B + e^S + e^T) \subseteq \mathcal L(e^*)$. 

To prove the converse inclusion, we assume that $\seq v \in \mathcal L(e^*)$ and we show that $\seq v \in \mathcal L(e^B + e^S + e^T)$. By the semantics of $e^*$, $\seq v = (u_{f(0)}u_2 \ldots u_{f(1)-1}, u_{f(1)} \ldots u_{f(2)-1}, \ldots)$, for a word sequence $\seq u \in \mathcal L(e)$ and an unbounded and nondecreasing function $f : \mathbb N
\rightarrow \bbNp$, with $f(0) = 1$.

Let $N(\seq v)$ be the sequence of exponents $(n_1,n_2, \ldots)$. If $N(\seq v)$ is bounded, then $\seq v \in \mathcal L(e^B) \subseteq \mathcal L(e^B + e^S + e^T)$. Otherwise, let $I = \langle i_1,i_2, \ldots \rangle$ be the increasing
sequence of indexes $i_j$ such that $n_{i_j} \in N_{i}(\seq v)$ %for every $i_j \in I$
and $F = \langle f_1,f_2, \ldots \rangle$ be the increasing sequence of indexes $f_j$ such that $n_{f_j} \in N_{f}(\seq v)$. %for every $f_j \in F$.
It clearly holds that $I \cup F = \bbNp$.
Now, let $\seq t = (t_1, t_2, \ldots)$ be the word sequence such that $t_{f_j} =
v_{f_j}$ for every $f_j \in F$ and $t_{i_j} = u_1 \ldots u_j$ for every $i_j
\in I$.
Clearly, $\seq t \in \mathcal L(e^S)$.
Moreover, let $\seq w = (w_1, w_2, \ldots)$ be the word sequence such that
$w_{i_j} = v_{i_j}$ for every $i_j \in I$ and $w_{f_j} = v_{i_1}$ for every $f_j
\in F$.
If $N_i(\seq v)$ is finite, then $N(\seq w)$ is bounded by
$\max(N_i(\seq v))$, and thus $\seq w \in \mathcal L(e^B)$; otherwise,
$N_i(\seq w)$ ($= N_i(\seq v)$) is infinite, that is, there are infinitely many
exponents in $\seq w$ occurring infinitely often, and thus $\seq w \in
\mathcal L(e^T)$.
Hence, $\seq w \in \mathcal L(e^B) \cup \mathcal L(e^T) \subseteq \mathcal L(e^B
+ e^T)$.
Since $\seq v$ is such that $v_k = t_k$, if $k \in F$, and $v_k = w_k$, if $k \in
I$, 
%we have that 
$\seq v \in \mathcal L(e^B + e^S + e^T)$.
\end{proof}

%, where $n_i = f(i) - f(i-1)$ for every $i \in \mathbb N^+$.
%
%Otherwise, we partition $N$ into $N_{i}$, containing those elements that occur infinitely often in $N$, and $N_{f}$, 
%containing those elements that occur only finitely many times in $N$.

%\dariomarginpartodo{remove this completely?}
%In view of Proposition \ref{prop:star-definition}, we conjecture the class obtained by extending $\omega{BS}$-regular languages with the $T$-constructor to be closed under complementation. By the results given in~\cite{DBLP:conf/stacs/BojanczykPT16}, we know that any extension of the class of $\omega{BS}$-regular languages that is closed under complementation (and that preserves the other closure properties exhibited by
%$\omega{BS}$-regular languages) turns out to be undecidable.
%However, we believe it possible to recover decidability by restricting to suitable subclasses of $\omega{BST}$-regular languages, e.g., $\omega{BT}$- and $\omega{ST}$-regular languages, with possibly weakened closure properties.

%thus making $\omega{T}$-regular languages worth being studied.
%\marginpar{\tiny La frase finale della footnote (``However \ldots'') non mi
%piace tanto, almeno cosi come e' scritta, ma l'ho inserita comunque perche'
%suggerita a tempi da Angelo nella risposta ai revisori di LICS.
%Rischiamo di darci la zappa sui piedi dicendo che vogiamo trovare la chiusura
%e poi limitarla per riottenere decidibilita'.
%Dobbiamo pensare a come dirla in una maniera diversa e meno rischiosa.}

%

%%% Local Variables:
%%% TeX-master: "main"
%%% End:

\section{Counter-check automata}\label{sec:automata}

%\begin{figure}[t]
\begin{wrapfigure}[12]{r}[0pt]{60mm}

%\vspace{-22mm}

\centering

%\hspace{-4mm}
%\parbox{52.5mm}{
%\includegraphics[scale=0.65]{automa}
\includegraphics[scale=0.53]{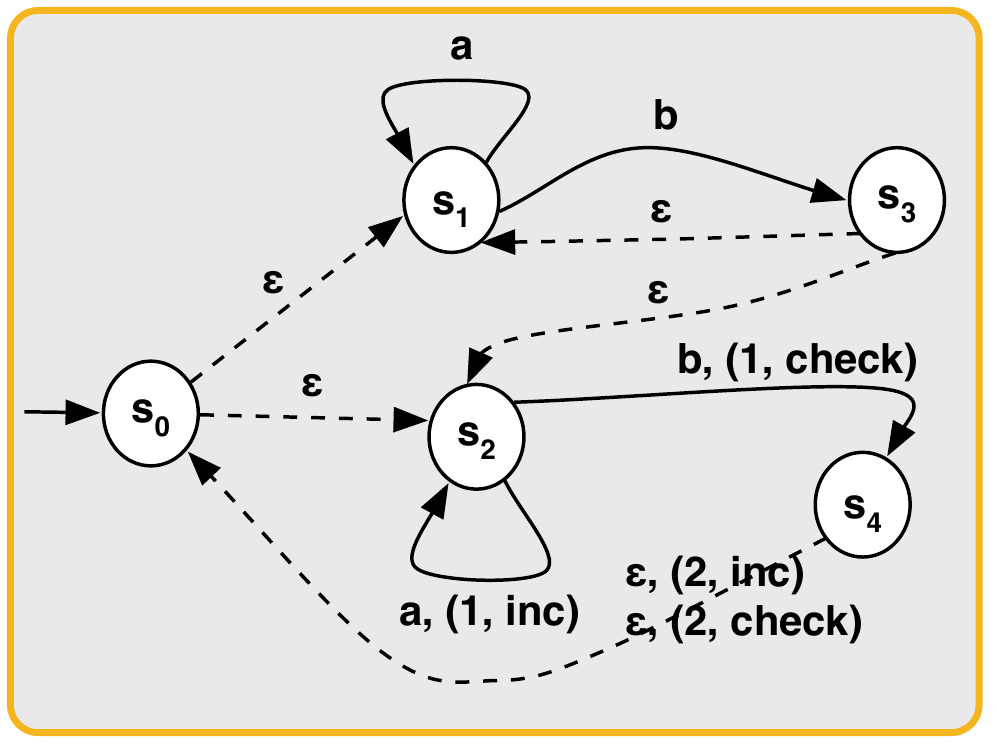}

\vspace{-3mm}

\caption{\label{fig:automata}A \cca for
  the language $((a^*b)^*a^Tb)^\omega$ ($N=2$).}
%}
\end{wrapfigure}
%\end{figure}
%

In this section, we introduce a new class of automata, called counter-check
automata, and we show that their emptiness problem is
decidable in PTIME.
In the next section, we will show that they are expressive enough to encode $\omega{T}$-regular expressions.

%\subsection{Definition}
%

A \emph{counter-check automaton} (an example is given in \figurename{}~\ref{fig:automata}) is an automaton equipped with a fixed number of counters. A transition can possibly \emph{increment} or \emph{reset} one of them (or do nothing). We refer to reset operations as \emph{check} operations to put the emphasis on the fact that computations keep trace of the evolution of the counter values. In particular, the acceptance condition depends on the sequences of \emph{check values} (i.e, the values when a check operation is performed) for all counters.

\begin{definition}[\cca]
  A \emph{counter-check automaton}  (\emph{\cca} for short) is a quintuple ${\cal A} = (S,\Sigma,s_0,N,\Delta)$, where $S$ is a finite set of states, $\Sigma$ is a finite alphabet, $s_0\in S$ is the initial state, $N \in \bbNp$ is the number of counters, and $\Delta\subseteq S\times(\Sigma \cup \{\epsilon\}) \times S \times  (\{1,\ldots,N\} \times \{no\_op, inc, check\})$ is a transition relation, subject to the constraint: if $(s,\sigma,s', (k,op)) \in \Delta$ and $op = no\_op$, then $k = 1$.
\end{definition}

A \emph{configuration} of a \cca ${\cal A}=(S,\Sigma,s_0,N, \Delta)$ is a pair $(s, \cvector)$, where $s \in S$ and $\cvector \in \bbN^N$ is called \emph{counter vector}. For $\cvector \in \bbN^N$ and $i \in \{ 1, \ldots, N\}$, let $\cvectorvaluevi$ be the $i$-th component of \cvector, i.e., the value of the $i$-th counter.

Let ${\cal A}=(S,\Sigma,s_0,N,\Delta)$ be a \cca.
% be a $CQ$ automaton.
We define a ternary relation $\rightarrow_{\cal A}$ over pairs of configurations and symbols in $\Sigma \cup \{\epsilon\}$ such that for all configuration pairs $(s, \cvector),(s', \cvector')$ and $\sigma\in \Sigma \cup \{\epsilon\}$, $(s, \cvector) \rightarrow^{\sigma}_{\cal A} (s', \cvector')$ iff there is $\delta= (s,\sigma,s',$ $(k,op)) \in \Delta$ such that $\cvectorvalue[v']{h}=\cvectorvalue{h}$ for all $h\neq k$, and
\begin{compactitem}
\item  if $op=no\_op$, then $\cvectorvalue[v']{k}=\cvectorvalue{k}$;
\item if $op=inc$, then $\cvectorvalue[v']{k}=\cvectorvalue{k} + 1$;
\item if $op = check$, then  $\cvectorvalue[v']{k}=0$.
\end{compactitem}
In such a case, we say that $(s, \cvector) \rightarrow^{\sigma}_{\cal A} (s', \cvector')$ via $\delta$.
Let $\rightarrow^{*}_{\cal A}$ be the reflexive and transitive closure of $\rightarrow^{\sigma}_{\cal A}$
(where we abstract away symbols in $\Sigma \cup \{ \epsilon\}$).
%
%defined as: (i)  $(s, C) \rightarrow^{*}_{\cal A} (s, C)$ for all configurations $(s,C)$; (ii) for every  pair of configurations $(s,C), (s', C')$ such that $(s,C)\rightarrow^{*}_{\cal A}(s', C')$  if there exist $\sigma \in \Sigma \cup \{ \epsilon\}$ and $(s'', C'')$ with $(s',C')\rightarrow^{\sigma}_{\cal A} (s'', C'')$  then $(s,C)\rightarrow^{*}_{\cal A}(s'', C'')$.
%% and $(s',C')\rightarrow^{*}_{\cal A}(s'', C'')$.
%
The \emph{initial configuration} of $\cal A$ is the pair $(s_0,\cvectorinitial)$, where for each $k \in \{ 1, \ldots, N \}$ we have $\cvectorinitialvalue{k} = 0$. A \emph{computation} of ${\cal A}$ is an infinite sequence of configurations $\cC=(s_0,\cvectorinitial)(s_1, \cvectorvsub{1})\ldots$, where, for all $i \in \mathbb N$,
%$(s_0,C_0)$ is the initial configuration and
$(s_i,\cvectorvsub{i}) \rightarrow^{\sigma_i}_{\cal A}(s_{i+1}, \cvectorvsub{i+1})$ for some $\sigma_i \in \Sigma\cup\{\epsilon\}$ (see Figure~\ref{fig:computation}).
For a computation $\cC =(s_0,\cvectorinitial)(s_1, \cvectorvsub{1})\ldots$ we
let \infcheckCk ($k \in \{ 1, \ldots, N \}$) denote the set $\{ n \in \bbN \mid
\forall h \exists i > h $ such that $\cvectorvaluevsubik = n$ and
$\cvectorvaluevsub{i+1}{k} = 0\}$, that is, \infcheckCk is the set of values of
the $k$-th counter that are checked infinitely often along $\cC$.
Given two configurations $(s_i,\cvectorvsub{i})$ and $(s_j,
\cvectorvsub{j})$ in $\cC$, with $i\leq j$, we say that $(s_j, \cvectorvsub{j})$
is \emph{$\epsilon$-reachable} from $(s_i,\cvectorvsub{i})$, written
$(s_i,\cvectorvsub{i}) \rightarrow^{*\epsilon}_{\cal A}(s_j, \cvectorvsub{j})$,
if $(s_{j'-1},\cvectorvsub{j'-1}) \rightarrow^{\epsilon}_{\cal A}(s_{j'},
\cvectorvsub{j'})$ for all $j' \in \{i+1, \ldots, j \}$.

\begin{figure*}[t!]
{

\hspace{.7mm}
\includegraphics[scale=0.663]{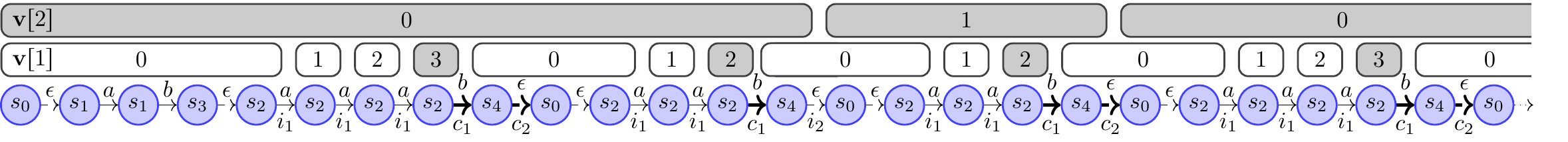}

}

\caption{\label{fig:computation}A prefix of a computation of the automaton in Figure \ref{fig:automata}.
A configuration is characterised by a circle (state) and the rounded-corner rectangles above it (counter vector). $\cvectorvaluevi$  is a counter vector component. Checked values for counters are highlighted in gray, with the corresponding transitions being written in boldface.}
%s of counter-queue configurations.}
\end{figure*}

A \emph{run} \phantomsection\label{txt:run-definition}
$\pi$ of $w$ on $\mathcal A$ is a computation $\pi = (s_0,\cvectorinitial)(s_1,\cvectorvsub{1}) \ldots$ for which there exists an increasing function $f: \bbNp \rightarrow \bbN $, called \emph{trace of $w$ in $\pi$ wrt.~$\cA$}, such that:
\begin{compactitem}
\item $(s_0, \cvectorinitial) \rightarrow^{*\epsilon}_{\cal A}
  (s_{f(1)},\cvectorvsub{f(1)})$, and
\item for all $i\geq 1$, $(s_{f(i)},\cvectorvsub{f(i)})\rightarrow^{w[i]}_{\cal A}
  (s_{f(i)+1}, \cvectorvsub{f(i)+1})$ and $(s_{f(i)+1}, \cvectorvsub{f(i)+1})
  \rightarrow^{*\epsilon}_{\cal A} (s_{f(i+1)},\cvectorvsub{f(i+1)})$.
\end{compactitem}

\label{txt:accepting-run-definition}
\noindent
A run $\pi = (s_0,\cvectorinitial)(s_1,\cvectorvsub{1}) \ldots$ of $w$ on
$\mathcal A$ is \emph{accepting} iff $| \infcheckpik | = +\infty$ for every $k \in \mathcal \{ 1, \ldots, N \}$.
An $\omega$-word $w \in \Sigma^{\omega}$ is \emph{accepted} by ${\cal A}$ iff
there exists an accepting run of $w$ on $\mathcal A$; we denote by $\cL({\cal A}
)$ the set of all $\omega$-words $w\in \Sigma^\omega$ that are accepted by
${\cal A}$, and we say that $\cal A$ \emph{accepts} the language $\cL({\cal
  A})$.
As an example, Figure~\ref{fig:automata} depicts a \cca with two counters ($N = 2$) accepting the
language $((a^*b)^*a^Tb)^\omega$. (Note that an automaton for the same language with one
counter only can be devised as well.)

%\vspace{-1cm}
%%% Local Variables:
%%% TeX-master: "main"
%%% End:

\newcommand{\oee}{\overline{e}}
\newcommand{\ob}{\overline{b}}
\newcommand{\oj}{\overline{j}}

\newcommand{\hb}{\hat{b}}
\newcommand{\he}{\hat{e}}
\renewcommand{\wp}{\widecheck{p}}
\newcommand{\wb}{\widecheck{b}}
\newcommand{\we}{\widecheck{e}}
\newcommand{\oC}{\overline{C}}
\newcommand{\wpe}{\widecheck{pe}}
\newcommand{\wpb}{\widecheck{pb}}

\newcommand{\maxinc}[1]{inc\_max_{#1}}
\newcommand{\maxcounter}[1]{counter\_max_{#1}}

\subsection{Decidability of the emptiness problem} \label{sec:non-empty}

We now prove that the emptiness problem for \cca is decidable in PTIME. The proof consists of 3 
%main 
steps:  (i) we replace general \cca by simple ones;
%(ii) we associate a two-player game with each simple $CQ$ automaton;
(ii) we prove that their emptiness can be decided by checking the existence of finite witnesses of accepting runs;
(iii) we show that the latter 
%condition 
can be verified by checking for emptiness a suitable NFA.

%\paragraph{Simple \cca.}
\noindent{\bf Simple \cca.}
A \cca ${\cal A} = (S,\Sigma,s_0,N,\Delta)$ is \emph{simple} iff for each $s\in S$ either $|\{ (s,\sigma, s', (k,op))\in \Delta \}| = 1$ or $op=no\_op$, $k=1$, and $\sigma=\epsilon$  for all  $(s,\sigma, s', (k,op)) \in \Delta$.
Basically, a simple \cca has states of two kinds: those in which it can fire exactly one action and those in which it makes a nondeterministic choice.
Moreover, for all pairs of  configurations $(s,\cvectorv),(s',\cvectorvprime)$
% of a simple $CQ$ automaton ${\cal A} = (S,\Sigma,s_0,N,\Delta)$
with $(s,\cvectorv)\rightarrow^{\sigma}_{\cA}(s',\cvectorvprime)$, the transition $\delta\in \Delta$ that has been fired in $(s,\cvectorv)$
is uniquely determined by
%the states
$s$ and $s'$. By exploiting \emph{$\epsilon$-transitions}, that is, transitions of the form $(s,\epsilon, s', (k,op))$, and by adding a suitable number of states, it can be easily shown that every \cca ${\cal A}$ may be turned into a simple one ${\cal A}'$ such that $\cL(\cA)=\cL(\cA')$.
Without loss of generality, in the rest of the section we restrict our attention
to simple \cca.

The set of states of a \cca
can be partitioned in four subsets:
\begin{inparaenum}[$(i)$]
\item the set of states $s$ from which only one transition of the form $(s,\sigma, s', (k,check))$  can be fired ($check_k$ states);
\item the set of states $s$ from which only one transition of the form $(s,\sigma, s', (k,inc))$ can be fired
($inc_k$ states);
\item the set of states $s$ from which only one transition of the form $(s,\sigma, s', (1,no\_op))$, with $\sigma \neq \epsilon$, can be fired ($sym$ states);
\item the set of states $s$ from which possibly many transitions of the form $(s,\epsilon, s', (1,no\_op))$ can be fired
($choice$ states).
\end{inparaenum}

Let ${\cal A} = (S,\Sigma,s_0,N,\Delta)$ be a \cca.
%A \emph{partial computation} of $\cA$ is a \emph{finite sequence}
%$\cP=(s'_0, \cvectorvprimesub{0}) \ldots (s'_n,\cvectorvprimesub{n})$ such that, for all
%$i \in \{ 1, \ldots, n-1 \}$,
%$(s'_i, \cvectorvprimesub{i}) \rightarrow^{\sigma_i}_{\cA} (s'_{i+1}, \cvectorvprimesub{i+1})$, for some
%$\sigma_i \in \Sigma\cup \{\epsilon\}$.
%%
%If $(s'_0,\cvectorvprimesub{0}) = (s_0,\cvectorinitial)$, i.e., $(s'_1,\cvectorvprimesub{0})$
%is the initial configuration of $\cA$, then $\cP$ is
%a \emph{prefix computation} of $\cA$.
A \emph{prefix computation} of $\cA$ is a finite prefix of a computation of
$\cA$;
formally, it is a finite sequence
$\cP=(s_0, \cvectorvsub{0}) \ldots (s_n,\cvectorvsub{n})$ such that, for all
$i \in \{ 0, \ldots, n-1 \}$,
$(s_i, \cvectorvsub{i}) \rightarrow^{\sigma_i}_{\cA} (s_{i+1}, \cvectorvsub{i+1})$, for some
$\sigma_i \in \Sigma\cup \{\epsilon\}$.
We denote by $\PrefixesA$ the sets of all
prefix computations of $\cA$.
%
%Clearly, $\PrefixesA \subseteq \PartialA$.
% holds.
%
%Given a prefix computation
%$\cP=(s_0,\cvectorinitial)\ldots(s_n,\cvectorvsub{n})$ and a partial computation
%$\cP' =(s'_0, \cvectorvprimesub{0})\ldots(s'_m,\cvectorvprimesub{m})$, we say
%that \emph{$\cP $ can be extended with $\cP'$} iff $\cP''=\cP \cdot \cP' =(s_0,
%\cvectorinitial)\ldots(s_n,\cvectorvsub{n}) (s'_0, \cvectorvprimesub{0})\ldots(s'_m, \cvectorvprimesub{m})$ is a prefix computation of $\mathcal A$
%($\cP''$ is said to be an \emph{extension} of $\cP$).
%%In such a case, we say that $\cP''$ is an extension of $\cP$.
%
For every prefix computation $\cP=(s_0,
\cvectorvsub{0}) \ldots (s_n,\cvectorvsub{n}) \in \PrefixesA$ and
$s \in S$, it holds that if $(s_n,
\cvectorvsub{n})\rightarrow^\sigma_{\cA} (s, \cvectorv)$, for some counter
vector \cvector and some $\sigma \in \Sigma \cup \{ \epsilon \}$, then \cvectorv
is uniquely determined by $s_n$, $\cvectorvsub{n}$, and $s$, that is,
there is no $\cvectorvprime \neq \cvectorv$ such that $(s_n,
\cvectorvsub{n})\rightarrow^{\sigma'}_{\cA} (s, \cvectorvprime)$, for any
$\sigma'$.
%We define the \emph{extension of $\cP$ with $s$}, denoted by $\cP << s$,  as
%$(s_0, C_0) \ldots (s_n,C_n)(s,C)$.

%\paragraph{Finite witnesses of accepting runs.}
\noindent{\bf Finite witnesses of accepting runs.}
We show now how to decide \cca emptiness by making use of the notion of
accepting witness for a \cca.

\begin{definition}[Accepting witness] \label{def:winningwitness}
  Let $\cA = (S,\Sigma,s_0,N,\Delta)$ be a \cca.
  A prefix computation $\cP=(s_0, \cvectorinitial)\ldots \allowbreak
  (s_n,\cvectorvsub{n})
  \in \PrefixesA$ is an \emph{accepting witness} (for $\mathcal A$) iff there
  are $2N+2$ indexes $\indexbegin < b_1 < e_1 < \ldots < b_N < e_N <
  \indexend$ such that $0 \leq \indexbegin$, $\indexend \leq n$, and the following conditions hold:
%
% Basically we are looking for a finitely representable witness for a winning
% strategy $str$.
%
\begin{compactenum}
\item \label{item:witness-nonepsilon}
  a non-$\epsilon$-transition can be fired from $s_{\indexbegin}$;
\item \label{item:witness-be}
  $s_{\indexbegin} = s_{\indexend}$ and, for each $k \in \{ 1, \ldots, N \}$,
  $s_{b_k} = s_{e_k}$, $s_{b_k}$ is an $inc_k$ state, and
  $s_j$ is not a $check_k$ state for any $j$ with $b_k \leq j \leq e_k$;
\item \label{item:witness-check}
  for each $k \in \{ 1, \ldots, N \}$, there is $j$, with $e_N < j <
  \indexend$,
  such that $s_j$ is a $check_k$ state.
%\item let $J = \{ j \in \bbN \mid 0 \leq j \leq \indexlimit$ and $s_j$ is a
%  $check_k$
%  state for some $k \}$; there is a set of indexes $\overline{J} = \{ \ob_j,
%  \oee_j \mid j \in J \}$ such that for all $j \in J$, with $s_{j}$ a $check_k$
%  state,
%(i) $\indexlimit < \ob_{j} < \oee_{j} < \indexend$,
%(ii) either $\oee_{j} < \ob_{j'}$ or $\oee_{j'} < \ob_{j}$, for all
%$j' \neq j$,
%(iii) $s_{\ob_{j}}$ and $s_{\oee_{j}}$ are
%$check_k$ states and there is no $check_k$ state in between them, and
%(iv) there are exactly $\cvectorvaluevsub{j}{k}$ many $inc_k$ states in between
%$s_{\ob_{j}}$ and $s_{\oee_{j}}$.
\end{compactenum}
\end{definition}
\noindent
An accepting witness for $\mathcal A$ can be seen as a finite representation of
an accepting run of some $\omega$-word on $\mathcal A$.
Thus, deciding whether a \cca $\mathcal A$ accepts the empty language amounts to
searching $\PrefixesA$ for accepting witnesses.
(The proof of the next lemma is omitted for lack of space.)
%
% given in Appendix~\ref{app:proof-winning}.)
%
%
\newcommand{\lemmawinning}{
  Let ${\cal A}$ be a \cca. Then, $\mathcal L(\mathcal A) \neq
  \emptyset$ iff $\PrefixesA$ contains an accepting witness.
}
\newcounter{lemmawinningcounter}
\setcounter{lemmawinningcounter}{\thetheorem}
\begin{lemma}\label{lemma:winning}
\lemmawinning
\end{lemma}

%\paragraph{From \cca to NFA.}
\noindent{\bf From \cca to NFA.}
Thanks to Lemma~\ref{lemma:winning}, deciding the emptiness problem for a \cca $\cA$ amounts to searching
$\PrefixesA$ for an accepting witness. Since we restricted ourselves to
%only consider 
simple \cca, we can safely identify elements of $\PrefixesA$ with their sequence of states and thus, by slightly abusing the notation, we can write, e.g., $s_0 s_1 \ldots s_n \in \PrefixesA$ for $(s_0, \cvectorinitial)\ldots(s_n,\cvectorvsub{n}) \in \PrefixesA$.
Given a \cca $\cA$, let \languagewitnessA be the language
of finite words over the alphabet $S$ (the set of states of $\cA$)
that are accepting witnesses for $\cA$.
It is easy to see that $\mathcal L(\cA) \neq \emptyset$ if and only if
$\languagewitnessA \neq \emptyset$.
In what follows, for a \cca $\cA$ we build a nondeterministic finite automata
(NFA) whose language is exactly $\languagewitnessA$.
Since the emptiness problem for NFA is decidable, so is the one for \cca.

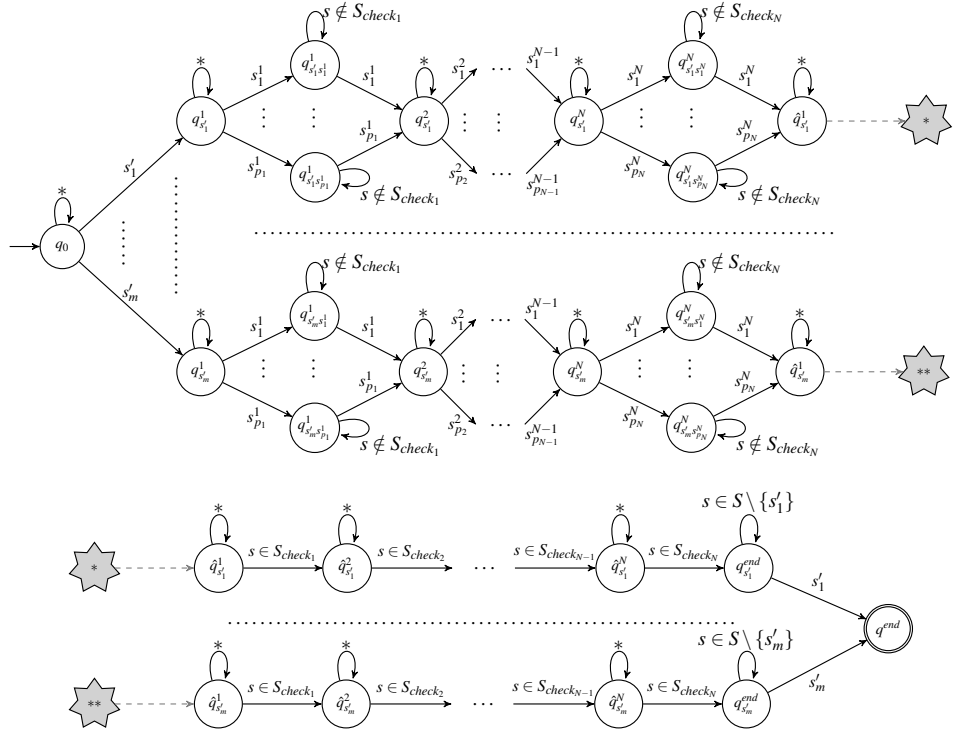
\begin{figure*}[t]
  \centering

\scalebox{.7}{%
  \begin{tikzpicture}
%    [->,>=stealth',shorten >=1pt,auto,node distance=2.1cm,semithick]
    [->,>=stealth',semithick]
    \tikzstyle{mystate}=[draw,circle,text width=11mm,align=center,scale=.75,
    inner sep=0mm]
    \tikzstyle{mystar}=[draw,star,star points=7,text width=4mm,align=center,
    scale=.7,fill=gray!35]

    \node[mystate] (q0) {$q_0$};
    \node[mystate,above right=17.5mm and 20mm of q0] (q1s1) {$q^1_{s'_1}$};
    \node[mystate,below right=17.5mm and 20mm of q0] (q1sm) {$q^1_{s'_m}$};

    \node[mystate,above right=4mm and 15mm of q1s1] (q1s1s1)
    {$q^1_{s'_1s^1_1}$};
    \node[mystate,below right=4mm and 15mm of q1s1] (q1s1sp)
    {$q^1_{s'_1s^1_{p_1}}$};

    \node[mystate,above right=4mm and 15mm of q1sm] (q1sms1)
    {$q^1_{s'_ms^1_1}$};
    \node[mystate,below right=4mm and 15mm of q1sm] (q1smsp)
    {$q^1_{s'_ms^1_{p_1}}$};

    \node[mystate,right=33mm of q1s1] (q2s1) {$q^2_{s'_1}$};
    \node[mystate,right=33mm of q1sm] (q2sm) {$q^2_{s'_m}$};

    \node[mystate,right=20mm of q2s1] (qNs1) {$q^N_{s'_1}$};
    \node[mystate,right=20mm of q2sm] (qNsm) {$q^N_{s'_m}$};

    \node[mystate,above right=4mm and 15mm of qNs1] (qNs1s1)
    {$q^N_{s'_1s^N_1}$};
    \node[mystate,below right=4mm and 15mm of qNs1] (qNs1sp)
    {$q^N_{s'_1s^N_{p_N}}$};

    \node[mystate,above right=4mm and 15mm of qNsm] (qNsms1)
    {$q^N_{s'_ms^N_1}$};
    \node[mystate,below right=4mm and 15mm of qNsm] (qNsmsp)
    {$q^N_{s'_ms^N_{p_N}}$};

    \node[mystate,right=33mm of qNs1] (hatq1s1a) {$\hat q^1_{s'_1}$};
    \node[mystate,right=33mm of qNsm] (hatq1sma) {$\hat q^1_{s'_m}$};

    \node[mystar,right=15mm of hatq1s1a] (star1a) {$*$};
    \node[mystar,right=15mm of hatq1sma] (starma) {$**$};

    \node[mystar,below right=55mm and 0mm of q0] (star1) {$*$};
    \node[mystar,below=17.5mm of star1] (starm) {$**$};

    \node[mystate,right=15mm of star1] (hatq1s1) {$\hat q^1_{s'_1}$};
    \node[mystate,right=15mm of starm] (hatq1sm) {$\hat q^1_{s'_m}$};

    \node[mystate,right=15mm of hatq1s1] (hatq2s1) {$\hat q^2_{s'_1}$};
    \node[mystate,right=15mm of hatq1sm] (hatq2sm) {$\hat q^2_{s'_m}$};

    \node[mystate,right=42.5mm of hatq2s1] (hatqNs1) {$\hat q^N_{s'_1}$};
    \node[mystate,right=42.5mm of hatq2sm] (hatqNsm) {$\hat q^N_{s'_m}$};

    \node[mystate,right=15mm of hatqNs1] (qends1) {$q^{\indexend}_{s'_1}$};
    \node[mystate,right=15mm of hatqNsm] (qendsm) {$q^{\indexend}_{s'_m}$};

    \node[mystate,accepting,below right=5mm and 20mm of qends1]
    (qend) {$q^{\indexend}$};

    \path
    (q0)++(-1,0) edge (q0)

    (q0) edge [loop above] node[above=-1mm]{$*$} (q0)
    (q1s1) edge [loop above] node[above=-1mm]{$*$} (q1s1)
    (q1sm) edge [loop above] node[above=-1mm]{$*$} (q1sm)
    (q1s1s1) edge [loop above] node[right]{$s \notin S_{check_1}$} (q1s1s1)
    (q1s1sp) edge [loop right] node[below right=0mm and -3mm]
    {$s \notin S_{check_1}$} (q1s1sp)
    (q1sms1) edge [loop above] node[right]{$s \notin S_{check_1}$} (q1sms1)
    (q1smsp) edge [loop right] node[below right=0mm and -3mm]
    {$s \notin S_{check_1}$} (q1smsp)
    (q2s1) edge [loop above] node[above=-1mm]{$*$} (q2s1)
    (q2sm) edge [loop above] node[above=-1mm]{$*$} (q2sm)
    (qNs1) edge [loop above] node[above=-1mm]{$*$} (qNs1)
    (qNsm) edge [loop above] node[above=-1mm]{$*$} (qNsm)
    (qNs1s1) edge [loop above] node[right]{$s \notin S_{check_N}$} (qNs1s1)
    (qNs1sp) edge [loop right] node[below right=0mm and -3mm]
    {$s \notin S_{check_N}$} (qNs1sp)
    (qNsms1) edge [loop above] node[right]{$s \notin S_{check_N}$} (qNsms1)
    (qNsmsp) edge [loop right] node[below right=0mm and -3mm]
    {$s \notin S_{check_N}$} (qNsmsp)
    (hatq1s1a) edge [loop above] node[above=-1mm]{$*$} (hatq1s1a)
    (hatq1sma) edge [loop above] node[above=-1mm]{$*$} (hatq1sma)
    (hatq1s1) edge [loop above] node[above=-1mm]{$*$} (hatq1s1)
    (hatq1sm) edge [loop above] node[above=-1mm]{$*$} (hatq1sm)
    (hatq2s1) edge [loop above] node[above=-1mm]{$*$} (hatq2s1)
    (hatq2sm) edge [loop above] node[above=-1mm]{$*$} (hatq2sm)
    (hatqNs1) edge [loop above] node[above=-1mm]{$*$} (hatqNs1)
    (hatqNsm) edge [loop above] node[above=-1mm]{$*$} (hatqNsm)
    (qends1) edge [loop above] node[above=-1mm]{$s \in S \setminus \{ s'_1 \}$}
    (qends1)
    (qendsm) edge [loop above] node[above=-1mm]{$s \in S \setminus \{ s'_m \}$}
    (qendsm)

    (q0) edge node[above,font=\footnotesize]{$s'_1$} (q1s1)
    (q0) edge node[above,font=\footnotesize]{$s'_m$} (q1sm)

    (q1s1) edge node[above,font=\footnotesize]{$s^1_{1}$} (q1s1s1)
    (q1s1) edge node[below,font=\footnotesize]{$s^1_{p_1}$} (q1s1sp)
    (q1sm) edge node[above,font=\footnotesize]{$s^1_{1}$} (q1sms1)
    (q1sm) edge node[below,font=\footnotesize]{$s^1_{p_1}$} (q1smsp)

    (q1s1s1) edge node[above,font=\footnotesize]{$s^1_{1}$} (q2s1)
    (q1s1sp) edge node[above,font=\footnotesize]{$s^1_{p_1}$} (q2s1)
    (q1sms1) edge node[above,font=\footnotesize]{$s^1_{1}$} (q2sm)
    (q1smsp) edge node[above,font=\footnotesize]{$s^1_{p_1}$} (q2sm)

    (q2s1) edge node[above,font=\footnotesize]{$s^2_{1}$} ++(1,1)
    (q2s1) edge node[below,font=\footnotesize]{$s^2_{p_2}$} ++(1,-1)
    (q2sm) edge node[above,font=\footnotesize]{$s^2_{1}$} ++(1,1)
    (q2sm) edge node[below,font=\footnotesize]{$s^2_{p_2}$} ++(1,-1)

    (qNs1) ++(-1,1)  edge node[above=2mm,font=\footnotesize]{$s^{N-1}_{1}$}
    (qNs1)
    (qNs1) ++(-1,-1) edge node[below=2mm,font=\footnotesize]{$s^{N-1}_{p_{N-1}}$}
    (qNs1)
    (qNsm) ++(-1,1)  edge node[above=2mm,font=\footnotesize]{$s^{N-1}_{1}$}
    (qNsm)
    (qNsm) ++(-1,-1) edge node[below=2mm,font=\footnotesize]{$s^{N-1}_{p_{N-1}}$}
    (qNsm)

    (qNs1) edge node[above,font=\footnotesize]{$s^N_{1}$} (qNs1s1)
    (qNs1) edge node[below,font=\footnotesize]{$s^N_{p_N}$}  (qNs1sp)
    (qNsm) edge node[above,font=\footnotesize]{$s^N_{1}$} (qNsms1)
    (qNsm) edge node[below,font=\footnotesize]{$s^N_{p_N}$}  (qNsmsp)

    (qNs1s1) edge node[above,font=\footnotesize]{$s^N_{1}$} (hatq1s1a)
    (qNs1sp) edge node[above,font=\footnotesize]{$s^N_{p_N}$}  (hatq1s1a)
    (qNsms1) edge node[above,font=\footnotesize]{$s^N_{1}$} (hatq1sma)
    (qNsmsp) edge node[above,font=\footnotesize]{$s^N_{p_N}$}  (hatq1sma)

    (hatq1s1a) edge[dashed,gray] (star1a)
    (hatq1sma) edge[dashed,gray] (starma)

     (star1) edge[dashed,gray] (hatq1s1)
     (starm) edge[dashed,gray] (hatq1sm)

     (hatq1s1) edge node[above,font=\footnotesize]{$s \in S_{check_1}$} (hatq2s1)
     (hatq1sm) edge node[above,font=\footnotesize]{$s \in S_{check_1}$} (hatq2sm)

     (hatq2s1) edge node[above,font=\footnotesize]{$s \in S_{check_2}$} ++(2,0)
     (hatq2sm) edge node[above,font=\footnotesize]{$s \in S_{check_2}$} ++(2,0)

     (hatqNs1) ++(-2,0) edge node[above,font=\footnotesize]
     {$s \in S_{check_{N-1}}$} (hatqNs1)
     (hatqNsm) ++(-2,0) edge node[above,font=\footnotesize]
     {$s \in S_{check_{N-1}}$} (hatqNsm)

     (hatqNs1) edge node[above,font=\footnotesize]{$s \in S_{check_N}$} (qends1)
     (hatqNsm) edge node[above,font=\footnotesize]{$s \in S_{check_N}$} (qendsm)

     (qends1) edge node[above,font=\footnotesize]{$s'_1$} (qend)
     (qendsm) edge node[below,font=\footnotesize]{$s'_m$} (qend);

     \node[below right=2mm and 7mm of q0,rotate=90]{\ldots\ldots};
     \node[below right=7.5mm and 17mm of q0,rotate=90]
     {\ldots\ldots\ldots\ldots\ldots};
     \node[right] at (3.5,0.25) {
       \ldots\ldots\ldots\ldots\ldots\ldots\ldots\ldots\ldots\ldots\ldots\ldots
       \ldots\ldots\ldots\ldots\ldots\ldots\ldots\ldots\ldots\ldots\ldots\ldots
     };

     \node[right] at (3,-7.15) {
       \ldots\ldots\ldots\ldots\ldots\ldots\ldots\ldots\ldots\ldots\ldots\ldots
       \ldots\ldots\ldots\ldots\ldots\ldots\ldots\ldots\ldots
     };

     \node[above right=4.5mm and -2mm of q1s1sp,rotate=90]{\ldots};
     \node[above right=4.5mm and -11.5mm of q1s1sp,rotate=90]{\ldots};

     \node[above right=4.5mm and -2mm of qNs1sp,rotate=90]{\ldots};
     \node[above right=4.5mm and -11.5mm of qNs1sp,rotate=90]{\ldots};

     \node[above right=4.5mm and -2mm of q1smsp,rotate=90]{\ldots};
     \node[above right=4.5mm and -11.5mm of q1smsp,rotate=90]{\ldots};

     \node[above right=4.5mm and -2mm of qNsmsp,rotate=90]{\ldots};
     \node[above right=4.5mm and -11.5mm of qNsmsp,rotate=90]{\ldots};

     \node[above right=5mm and 8mm of q2s1]{\ldots};
     \node[above right=-15mm and 8mm of q2s1]{\ldots};
     \node[above right=-7mm and 6mm of q2s1,rotate=90]{\ldots};
     \node[above right=-7mm and 12.5mm of q2s1,rotate=90]{\ldots};

     \node[above right=5mm and 8mm of q2sm]{\ldots};
     \node[above right=-15mm and 8mm of q2sm]{\ldots};
     \node[above right=-7mm and 6mm of q2sm,rotate=90]{\ldots};
     \node[above right=-7mm and 12.5mm of q2sm,rotate=90]{\ldots};

     \node[right=18mm of hatq2s1]{\ldots};
     \node[right=18mm of hatq2sm]{\ldots};

  \end{tikzpicture}
}%end scalebox

  \caption{A graphical account of the automaton $\cN_1$: $S_{\noneps} = \{ s'_1,
    s'_2, \ldots, s'_m \}$, $S_{inc_k} = \{ s^k_1, s^k_2, \ldots, s^k_{p_k} \}$ ($k
    \in \{ 1, \ldots, N\}$).}
  \label{fig:nfa}
\end{figure*}

In what follows, without loss of generality, we restrict our attention to
accepting witnesses for which the set of indexes required by item 3
of Definition~\ref{def:winningwitness} is ordered.
More precisely (we borrow the notation from
Definition~\ref{def:winningwitness}), we assume that there are $N$ indexes $c_1
< \ldots <c_N$, with $e_N < c_1$ and $c_N < \indexend$, such that $s_{c_k}$ is a
$check_k$ state, for each $k \in \{ 1, \ldots, N \}$ (this requirement
strengthens the one imposed by item 3 of Definition~\ref{def:winningwitness}).
Given a \cca $\cA$, it is easy to check that $\PrefixesA$ contains an accepting
witness, as specified by Definition~\ref{def:winningwitness}, if and
only if it contains one satisfying the additional ordering property above.
Thus, Lemma~\ref{lemma:winning} holds with respect to the new definition of
accepting witness as well.

Given a \cca $\cA$, we apply the following steps to build an NFA $\cN$ such that
$\languagenfaN = \languagewitnessA$:
\begin{inparaenum}[$(i)$]
\item we build an NFA $\cN_1$ accepting finite words over the set of states of
  $\cA$ that are potential accepting witnesses, i.e., they satisfy
  conditions~\ref{item:witness-nonepsilon}-~\ref{item:witness-check} of
  Definition~\ref{def:winningwitness} but they might not be prefix computations;
%
  % for a \cca (possibly different from
%  $\mathcal A$) over the same alphabet as $\cA$;
%
  in other words, such an automaton
  might as well accept words not belonging to $\PrefixesA$;
%
%, e.g., words that are
%  accepting witnesses for another \cca over the same alphabet as $\mathcal A$;
\item since $\PrefixesA$ is a regular language, thanks to closure properties
  of NFA, there exists an NFA $\cN$ whose language is
  $\languagenfa{\cN_1} \cap \PrefixesA = \languagewitnessA$.
\end{inparaenum}

Let $\cA = (S, \Si, s_0, N, \De)$ be a \cca.
We define $\cN_1 = \allowbreak \langle Q, \Si_{\cN_1}, \de, q_0, F \rangle$ as follows.
We set $\Si_{\cN_1} = S$, $F = \{ q^{\indexend} \}$; moreover, let $S_{\noneps}$
be the set of states of $S$ from which a non-$\epsilon$-transition can be fired,
and $S_{inc_k}$ be the sets of $inc_k$ states in $S$ ($k \in \{ 1, \ldots, N
\}$), we set $ Q = \{ q_0, q^{\indexend} \} \cup \{q^{\indexend}_{s'} \mid s'
\in S_{\noneps} \} \cup \bigcup_{k=1}^N \{q^k_{s'}, \hat q^k_{s'} \mid s' \in
S_{\noneps} \} \cup \bigcup_{k=1}^N \{q^k_{s's''} \mid s' \in S_{\noneps}, s'' \in
S_{inc_k} \}.$
%
%{\centering
%  $\begin{array}{l@{\hspace{0mm}}l}
%     Q = \{ q_0, q^{\indexend} \}
%     & \cup \{q^{\indexend}_{s'} \mid s' \in S_{sym} \} \cup \\
%     & \cup \bigcup_{k=1}^N \{q^k_{s'}, \hat q^k_{s'} \mid s' \in S_{sym} \} \cup \\
%     & \cup \bigcup_{k=1}^N \{q^k_{s's''} \mid s' \in S_{sym}, s'' \in
%       S_{inc_k} \}.
%   \end{array}$
%
%}
%
The transition relation $\de$ is described in Figure~\ref{fig:nfa}.
In particular, the automaton behaves as follows:
\begin{compactenum}
\item \label{a} it nondeterministically guesses index $\indexbegin$ when a
  symbol $s' \in S_{\noneps}$ is read; the next state
  $q^1_{s'}$ reached by $\cN_1$ stores the information about the state $s'$ of
  $\cA$ being read to check, at a later stage (when index $\indexend$ is
  guessed), that $s_{\indexbegin} = s' = s_{\indexend}$;

\item\label{b} similarly, for each $k \in \{ 1, \ldots, N \}$, it
  nondeterministically guesses indexes $b_k$ and $e_k$, when a symbol $s^k$
  corresponding to an $inc_k$ state (of $\cA$) is read; once again, the
  information about the state $s^k$ of $\cA$ being read is stored in the next
  state $q^k_{s's^k}$ reached by $\cN_1$, in order to check that the same state
  $s^k$ is read when $e_k$ is guessed ($s_{b_k} = s^k = s_{e_k}$); moreover, the
  automaton forces the absence of $check_k$ state in between indexes $b_k$ and
  $e_k$;
\item\label{c} it checks for the existence, after $e_N$, of $check_k$ states ($k
  \in \{ 1, \ldots, N \}$) in the desired order;
\item wait for the input symbol $s'$, that is, the same symbol read when
  $\indexbegin$ was guessed; when such a symbol is read, $\cN_1$ enters the
  final state $q^{\indexend}$.
\end{compactenum}

\noindent
Let $S_{\noneps}$ and $S_{inc_k}$ ($k \in \{1, \ldots, N \}$) be defined as
above and, in addition, let $S_{check_k}$ be the set of $check_k$ states in $S$.
We formally define $\de$ as follows:

\noindent
$\begin{array}{@{\hspace{0mm}}l@{\hspace{.4mm}}r@{\hspace{.4mm}}l@{\hspace{0mm}}}
   \de & = & \multicolumn{1}{@{\hspace{.4mm}}l@{\hspace{.5mm}}}{
           \{ (q_0, s, q_0) \mid s \in S \}
           \cup \{ (q_0, s', q^1_{s'}) \mid s' \in S_{\noneps} \} \cup \bigcup_{k=1}^N \{ (q^k_{s'}, s, q^k_{s'}) \mid s' \in S_{\noneps},
           s \in S \}
           } \\
%         & \multicolumn{2}{@{\hspace{1mm}}l}{
%           \cup \bigcup_{k=1}^N \{ (q^k_{s'}, s, q^k_{s'}) \mid s' \in S_{\noneps},
%           s \in S \} \cup
%           } \\
%
         & \cup & \multicolumn{1}{@{\hspace{.4mm}}l@{\hspace{.5mm}}}{
           \bigcup_{k=1}^N \{ (q^k_{s'}, s^k, q^k_{s's^k}) \hspace{-.5mm} \mid \hspace{-.5mm} s' \in S_{\noneps}, s^k \in S_{inc_k} \} \hspace{-.3mm} \cup \bigcup_{k=1}^N \{ (q^k_{s's^k}, s, q^k_{s's^k}) \hspace{-.5mm} \mid \hspace{-.5mm} s' \in S_{\noneps}, s^k \in S_{inc_k}, s \in S \setminus S_{check_k} \}
%
%s \in S \setminus S_{check_k} \} \cup
           } \\
%         & \cup \bigcup_{k=1}^N \{ (q^k_{s's^k}, s, q^k_{s's^k}) \mid
%         & s' \in S_{\noneps}, s^k \in S_{inc_k}, \\
         & \cup & \multicolumn{1}{@{\hspace{.4mm}}l@{\hspace{.5mm}}}{
            \bigcup_{k=1}^{N-1} \{ (q^k_{s's^k}, s^k, q^{k+1}_{s'}) \mid
           s' \in S_{\noneps}, s^k \in S_{inc_k} \} \cup \{ (q^N_{s's^N}, s^N, \hat q^{1}_{s'}) \mid s' \in S_{\noneps},
           s^N \in S_{inc_N} \} } \\
%         & \multicolumn{2}{@{\hspace{0mm}}l}{
%           \cup \bigcup_{k=1}^{N-1} \{ (q^k_{s's^k}, s^k, q^{k+1}_{s'}) \mid
%           s' \in S_{\noneps}, s^k \in S_{inc_k} \} \cup
%           } \\
%         & \multicolumn{2}{@{\hspace{0mm}}l}{
%           \cup \bigcup_{k=1}^N \{ (q^k_{s'}, s^k, q^k_{s's^k}) \mid s' \in S_{\noneps},
%           s^k \in S_{inc_k} \} \cup \bigcup_{k=1}^N \{ (q^k_{s's^k}, s, q^k_{s's^k}) \mid
%           s' \in S_{\noneps}, s^k \in S_{inc_k},  s \in S \setminus S_{check_k} \} \cup }\\
%%         & \cup \bigcup_{k=1}^N \{ (q^k_{s's^k}, s, q^k_{s's^k}) \mid
%%         & s' \in S_{\noneps}, s^k \in S_{inc_k}, \\
%         & &
%         & \multicolumn{2}{@{\hspace{0mm}}l}{
%           \cup \bigcup_{k=1}^{N-1} \{ (q^k_{s's^k}, s^k, q^{k+1}_{s'}) \mid
%           s' \in S_{\noneps}, s^k \in S_{inc_k} \} \cup
%           }} \\
         & \cup & \multicolumn{1}{@{\hspace{.4mm}}l@{\hspace{.5mm}}}{
            \bigcup_{k=1}^{N} \{ (\hat q^k_{s'}, s, \hat q^{k}_{s'}) \mid
           s' \in S_{\noneps}, s \in S \} \cup \bigcup_{k=1}^{N-1} \{ (\hat q^k_{s'}, s^k, \hat q^{k+1}_{s'}) \mid
           s' \in S_{\noneps}, s^k \in S_{check_k} \}
           } \\
%         & \multicolumn{2}{@{\hspace{0mm}}l}{
%           \cup \bigcup_{k=1}^{N} \{ (\hat q^k_{s'}, s, \hat q^{k}_{s'}) \mid
%           s' \in S_{\noneps}, s \in S \} \cup
%           } \\
         & \cup & \multicolumn{1}{@{\hspace{.4mm}}l@{\hspace{.5mm}}}{
             \{ (\hat q^N_{s'}, s^N, q^{\indexend}_{s'}) \mid
           s' \in S_{\noneps}, s^N \in S_{check_N} \} \cup \{ (q^{\indexend}_{s'}, s, q^{\indexend}_{s'}) \mid
           s' \in S_{\noneps}, s \in S \setminus \{ s' \} \}
           } \\
%         & \multicolumn{2}{@{\hspace{0mm}}l}{
%           \cup \{ (\hat q^N_{s'}, s^N, q^{\indexend}_{s'}) \mid
%           s' \in S_{\noneps}, s^N \in S_{check_N} \} \cup
%           } \\
         & \cup & \multicolumn{1}{@{\hspace{.4mm}}l@{\hspace{.5mm}}}{
           \{ (q^{\indexend}_{s'}, s', q^{\indexend}) \mid s' \in S_{\noneps} \}
           }
%           } \\
%         & \multicolumn{2}{@{\hspace{0mm}}l}{
%           \cup \{ (q^{\indexend}_{s'}, s', q^{\indexend}) \mid s' \in S_{\noneps} \}.
%           }
\end{array}
$

Since the size of $\cN_1$ is polynomial in the size of $\cA$ ($|Q| \leq 2 + 2
\cdot N \cdot |S| + N \cdot |S|^2 + |S|$), we have a polynomial reduction from
the emptiness problem for \cca to the one for NFA.
%, and thus the following result holds.

\begin{theorem}
  The emptiness problem for \cca is decidable in PTIME.
\end{theorem}

%%% Local Variables:
%%% mode: latex
%%% TeX-master: "main"
%%% End:

\section{From \texorpdfstring{$\omega T$}{omega T}-regular languages to \texorpdfstring{\cca}{CCA}}
\label{sec:encoding}

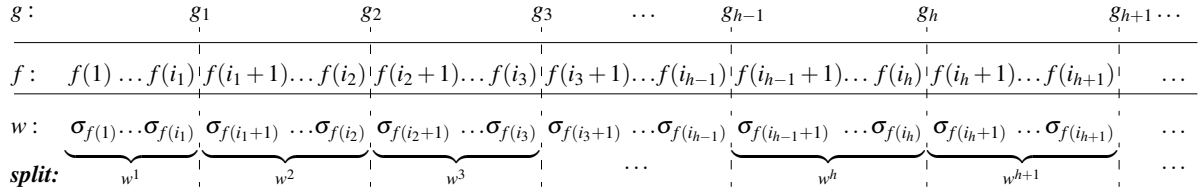
\begin{figure*}[t]
  \centering

\scalebox{.9}{%
  {\small $\begin{array}{@{\hspace{-.5mm}}l@{\hspace{-2mm}}l@{\hspace{0mm}}c@{\hspace{0mm}}r@{\hspace{0.5mm}}:@{\hspace{0mm}}l@{\hspace{0mm}}c@{\hspace{0mm}}r@{\hspace{0.5mm}}:@{\hspace{0mm}}l@{\hspace{0mm}}c@{\hspace{0mm}}r@{\hspace{0.5mm}}:@{\hspace{0mm}}l@{\hspace{0mm}}c@{\hspace{0mm}}r@{\hspace{0.5mm}}:@{\hspace{0mm}}l@{\hspace{0mm}}c@{\hspace{0mm}}r@{\hspace{0.5mm}}:@{\hspace{0mm}}l@{\hspace{0mm}}c@{\hspace{0mm}}r@{\hspace{0.5mm}}:@{\hspace{0mm}}c@{\hspace{-2.5mm}}}
     \boldsymbol{g:}
     & \multicolumn{3}{@{\hspace{0mm}}l@{\hspace{1mm}}}{}
     & \multicolumn{3}{@{\hspace{0mm}}l@{\hspace{1mm}}}{\hspace{-1.5mm}g_1}
     & \multicolumn{3}{@{\hspace{0mm}}l@{\hspace{1mm}}}{\hspace{-1.5mm}g_2}
     & \multicolumn{1}{@{\hspace{0mm}}l@{\hspace{1mm}}}{\hspace{-1.5mm}g_3}
     & \ldots
     & \multicolumn{1}{@{\hspace{0mm}}l@{\hspace{1mm}}}{}
     & \multicolumn{3}{@{\hspace{0mm}}l@{\hspace{1mm}}}{\hspace{-1.5mm}g_{h-1}}
     & \multicolumn{3}{@{\hspace{0mm}}l@{\hspace{1mm}}}{\hspace{-1.5mm}g_h}
     & \multicolumn{1}{@{\hspace{0mm}}l@{\hspace{1mm}}}{\hspace{-1.5mm}g_{h+1}
       \ldots} \\ [.25mm]
     &&&&&&&&&&&&&&&&&&& \\ [-2mm]
     \arrayrulecolor{black}\hline\arrayrulecolor{black}
     &&&&&&&&&&&&&&&&&&& \\ [-1.5mm]
     \boldsymbol{f:}
     & \hspace{.5mm} f(1) & \ldots & f(i_1)
     & \hspace{.5mm} f(i_1 + 1) & \ldots & f(i_2)
     & \hspace{.5mm} f(i_2 + 1) & \ldots & f(i_3)
     & \hspace{.5mm} f(i_3 + 1) & \ldots & f(i_{h-1})
     & \hspace{.5mm} f(i_{h-1} + 1) & \ldots & f(i_{h})
     & \hspace{.5mm} f(i_{h} + 1) & \ldots & f(i_{h+1})
     & \hspace{2mm} \ldots \\ [.5mm]
     \arrayrulecolor{black}\hline\arrayrulecolor{black}
     &&&&&&&&&&&&&&&&&&& \\ [-1.5mm]
     \boldsymbol{w:}
     & \hspace{.75mm} \sigma_{f(1)} & \ldots & \sigma_{f(i_1)}
     & \hspace{.75mm} \sigma_{f(i_1 + 1)} & \ldots & \sigma_{f(i_2)}
     & \hspace{.75mm} \sigma_{f(i_2 + 1)} & \ldots & \sigma_{f(i_3)}
     & \hspace{.75mm} \sigma_{f(i_3+1)}  & \ldots & \sigma_{f(i_{h-1})}
     & \hspace{.75mm} \sigma_{f(i_{h-1} + 1)} & \ldots & \sigma_{f(i_{h})}
     & \hspace{.75mm} \sigma_{f(i_{h} + 1)} & \ldots & \sigma_{f(i_{h+1})}
     & \hspace{2mm} \ldots \\ [-4mm]
%     \arrayrulecolor{black}\hline\arrayrulecolor{black}
%     &&&&&&&&&&&&&&&&&&& \\ [-1.5mm]
%      \fbox{\parbox{10mm}{\vspace{4mm} \mathit{\boldsymbol{split:}}}}
%     \parbox{10mm}{\vspace{6mm} \textit{\textbf{split:}}}
     \parbox[t][5mm][b]{10mm}{\textit{\textbf{split:}}}
     & \multicolumn{3}{@{\hspace{0mm}}r:@{\hspace{-1.5mm}}}
       {\underbrace{\hspace{19.5mm}}_{w^1}}
     & \multicolumn{3}{@{\hspace{0mm}}r:@{\hspace{-1.5mm}}}
       {\underbrace{\hspace{24.5mm}}_{w^2}}
     & \multicolumn{3}{@{\hspace{0mm}}r:@{\hspace{-1.5mm}}}
       {\underbrace{\hspace{24.5mm}}_{w^3}}
     & \multicolumn{3}{c:}{\parbox{14mm}{\centering \vspace{7.5mm} $\ldots$}}
     & \multicolumn{3}{@{\hspace{0mm}}r:@{\hspace{-1.5mm}}}
       {\underbrace{\hspace{28.5mm}}_{w^{h}}}
     & \multicolumn{3}{@{\hspace{0mm}}r:@{\hspace{-1.5mm}}}
       {\underbrace{\hspace{28mm}}_{w^{h+1}}}
%    & \multicolumn{1}{@{\hspace{0mm}}r:@{\hspace{-1.5mm}}}{\ldots}
     & \parbox{14mm}{\centering \vspace{7.5mm} \hspace{2mm}$\ldots$}
   \end{array}$}
}%end scalebox

  \begin{tikzpicture}

  \end{tikzpicture}
  \caption{An infinite word $w = w[1]w[2] \ldots w[i] \ldots =
    \sigma_{f(1)}\sigma_{f(2)} \ldots \allowbreak \sigma_{f(i)} \ldots$ is split
    using sequence $g_1 < g_2 < \ldots < g_h < \ldots$ into infinitely many
    finite words $w^1$, $w^2$, \ldots, $w^h$, \ldots.}
  \label{fig:word-splitting}
\end{figure*}

In this section, we show how to map an \emph{$\omega T$-regular expression} $E$
into a corresponding  \cca $\cA$ such that  $\cL(E)= \cL(\cA)$.
We build the automaton $\cA$ in a compositional way: for each sub-expression $E'$ of $E$,
starting from the atomic ones, we introduce a set $\mathcal S_{E'}$ of \ccas and
then we show how to produce the set of automata for complex sub-expressions by
suitably combining automata in the sets associated with their sub-expressions.
%according to the semantics of the operators of $\omega T$-regular expressions.
Eventually, we obtain a set of automata for the $\omega T$-regular expression $E$.
The automaton $\cA$ results from the merge of the automata in such a set,
as described below.
%W.l.o.g., we assume that the set of states of any automaton $\cA'$
%generated in such a construction is disjoint from the set
%of states of any other automaton $\cA''$ that comes into play in
%the construction (being $\cA''$ an element of the set $\cA'$ belongs
%to or of another one).
Without loss of generality, we assume the sets of states of all automata generated in the
construction to be pairwise disjoint, i.e., if $\cA' \in \mathcal S_{E'}$
and $\cA'' \in \mathcal S_{E''}$, where $E'$ and $E''$ are two (not necessarily distinct)
sub-expressions of $E$, then the set of states of $\cA'$ and the one of $\cA''$
are disjoint.

%W.l.o.g., for each sub-expression $T'$ of $T$ and every pair
%of distinct automata $\cA= (S,\Sigma,s_0,N,\Delta), \cA'=
%(S',\Sigma,s_0',N',\Delta') \in \mathcal S_{T'}$, we assume that
%$S\cap S'=\emptyset$.
%We also assume that $S\cap S'=\emptyset$ holds for every pair of expressions
%$T_1, T_2$ and every pair of automata
%$(S,\Sigma,s_0,N,\Delta)\in \mathcal S_{T_1},
%(S',\Sigma,s_0',N',\Delta')\in \mathcal S_{T_2}$.
%The construction is by structural induction on $\omega T$-regular expressions.
%By inductive hypothesis, when building the set $\mathcal S_{T'}$ of $CQ$
%automata for a sub-expression $T'$ of $T$, we will assume that we have
%the set $\mathcal S_{T''}$ for each sub-expression $T''$ of $T'$.
%Moreover, by construction, all the $CQ$ automata
%$\cA= (S,\Sigma,s_0,N,\Delta)$ we build feature a distinguished
%\emph{final state}, denoted by $s_f\in S$, such that
%$(s, \sigma, s', (k, op))\in \Delta$ implies $ s \neq s_f$.

We proceed by structural induction on $\omega T$-regular expressions,
that is, when building the set $\mathcal S_{E'}$ of \ccas for a
sub-expression $E'$ of $E$, we assume the sets of \ccas for
the sub-expressions of $E'$ to be available.
In addition, by construction, we force each generated \cca $\cA=
(S,\Sigma,s_0,N,\Delta)$ to feature a distinguished \emph{final state} $s_f$
such that $(s_f, \sigma, s', (k, op))\in \Delta$ implies $\sigma = \epsilon$,
$s' = s_f$, $k=1$, and $op = inc$; in order to distinguish the final state of a
\cca we sometimes abuse the notation and write $\cA =
(S,\Sigma,s_0,s_f,N,\Delta)$, where $s_f$ is the final state of $\cA$.

%\subsection{Encoding of \texorpdfstring{$T$}{T}-regular expressions}

%\paragraph{Encoding of \texorpdfstring{$T$}{T}-regular expressions.}
\noindent{\bf Encoding of \texorpdfstring{$T$}{T}-regular expressions.}
We first deal with $T$-regular expressions (sub-grammar rooted in $e$ in paragraph ``$\omega T$-regular languages'' at page~\pageref{page:Tregexp-grammar}).
Since a $T$-regular expression produces a language of word sequences and our automata accept $\omega$-words,
we must find a way to extract sequences from $\omega$-words.
Intuitively, we do that by splitting an infinite word into infinitely many
finite sub-words, each of them corresponding to the sequence of symbols
in between two consecutive \emph{check} of the 1st counter along the
corresponding accepting run.
Formally, let $\pi=(s_0,\cvectorinitial)(s_1, \cvectorvsub{1})\ldots$ be an
accepting run of some $\omega$-word $w$ on some \cca $\cA$ such that $(s_i,
\cvectorvsub{i}) \rightarrow^{\sigma_i}_{\cal A} (s_{i+1}, \cvectorvsub{i+1})$
via $\delta_i$, for each $i \geq 0$, and let $f$ be the trace of $w$ in $\pi$
wrt.~$\cA$ (see definition of run at page~\pageref{txt:run-definition}).
Recall that $f$ is such that $\sigma_{f(i)} = w[i]$ for all $i\geq 1$ (roughly
speaking, $f$ enumerates symbols different from $\epsilon$ within sequence
$\sigma_0 \sigma_1 \ldots$).
Moreover, let $g_1 < g_2 < \ldots < g_h < \ldots$ ($g_h \in \mathbb N$ for every
$h$) be the sequence of indexes corresponding to transitions in $\pi$ where the
1st counter is checked, that is, for every $i\in \bbN$ we have that $\delta_i$
has the form $(s_i, \sigma_i, s_{i+1},(1,check))$ if and only if $i = g_h$ for
some $h$.
As shown in \figurename{}~\ref{fig:word-splitting},
the sequence $\langle g_h \rangle_{h \in \bbNp}$ defines a unique partition of
the infinite word $w = \sigma_{f(1)} \sigma_{f(2)} \sigma_{f(3)} \ldots$
into infinitely many finite sub-words (some of them are possibly empty words):
$w^1 = \sigma_{f(1)}\sigma_{f(2)} \ldots \sigma_{f(i_1)}$, $w^2 = \sigma_{f(i_1+1)} \sigma_{f(i_1+2)} \ldots \sigma_{f(i_2)}$, $w^3 = \sigma_{f(i_2+1)} \ldots \sigma_{f(i_3)}$, \ldots, $w^h = \sigma_{f(i_{h-1}+1)} \ldots \sigma_{f(i_h)}$, and so on,
with $f(i_h) < g_h \leq f(i_h+1)$ for every $h$.
We define the language of word sequences accepted by
$\cA$, denoted by $\cL_s(\cA)$, as $\cL_s(\cA)=\{ (w^1,w^2, \ldots, w^h, \ldots): w \in \cL(\cA) \}$.

%Automata for $T$-regular expressions are built as follows.
%$t  \ ::= \ \emptyset \ \mid \ a \ \mid \ t.t \ \mid \ t+t
%\ \mid \ t^* \ \mid \ t^T$.
%
%For each expression $t$, we build a set $\cS_t=\{ \cA_1,\ldots, \cA_n \}$, with
%$\cA_i= (S_i,\Sigma,s^i_0,s^i_f,N_i,\Delta_i)$ and $i \in \{ 1, \ldots, n
%\}$, such that $\cL(t)$ equals to
%
%\noindent
%$\bigcup\limits_{1 \leq i \leq n} \hspace{-2.5mm} \cL_s
%((S_i, \Sigma,s^i_0,s^i_f,N_i,\Delta_i\cup\{ (s^i_f,\epsilon, s^i_0, (1,check))\}))$.
%

Let $\widehat{\cA} = (S, \Sigma,s_0,s_f,N,\Delta\cup\{ (s_f,\epsilon, s_0,
(1,check))\})$, for every $\cA = (S, \Sigma,s_0,s_f,N,\Delta)$.
For each expression $e$, we build a set $\cS_e$ for which it holds:

{\centering
\noindent
$\cL(e) = \bigcup_{\cA \in \cS_e} \cL_s (\widehat{\cA})$.

}

\noindent{\em Base cases.} If $e=\emptyset$, then $\mathcal S_{e}=\{\cA_{\emptyset}\}$ where
$\cA_{\emptyset}=(\{s_0,s_f\}, \Sigma, s_0, s_f, 1, \emptyset)$.

\noindent
If $e=a$, then $\mathcal S_{e} = \{\cA_a\}$
where $\cA_a= (\{s_0,  s_f\}, \allowbreak \Sigma, s_0, s_f, 1, \{ (s_0, a, s_f, (1, no\_op)), (s_f,\epsilon, s_f, (1, inc))\})$.

\noindent
See \figurename{}~\ref{fig:translation}~(a) and~(b)  for a graphical account of both cases.

%To deal with the inductive cases, we can assume, by inductive hypothesis, that
%we have the set $\mathcal S_{t'}$ for each sub-expression $t'$ of $t$.

\smallskip

\noindent{\em Inductive step.}
For our purposes, we define, for every \cca $\cA= (S,\Sigma,s_0,N,\Delta)$ and
natural number $N' \geq 1$, the \emph{$N'$-shifted version} of $\cA$ as the
automaton $\cA'= (S,\Sigma, s_0,N+N',\{ (s,\sigma, s,(k+N',op)): (s,\sigma,
s,(k,op)) \in \Delta \})$.
% Notice that all the counter indexes  are shifted but the first one.
%
Four cases must be considered.

\begin{figure}[!b]
\vspace{-0.9cm}
\centering
\scalebox{.69}{% 
\begin{tikzpicture}
\node[label={[]180:(a)}](A){
\includegraphics{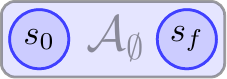}};

\pgftransformshift{\pgfpoint{0.8cm}{-1.5cm}}

\node[label={[]180:(b)}](B){\includegraphics{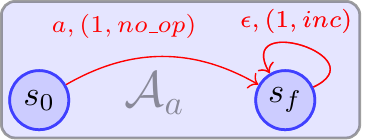}};

\pgftransformshift{\pgfpoint{1.9cm}{-1.5cm}}

\node[label={[yshift=-0.4cm]180:(c)}](C){\includegraphics{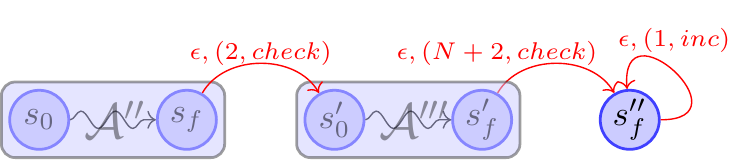}};
\draw (C)++(-3.95,.32) node[draw=gray!30,fill=gray!10,font=\footnotesize]{$\mathbf{e_1 \cdot e_2}$};

\pgftransformshift{\pgfpoint{11cm}{2.5cm}}

\node[label={[yshift=-0.2cm, xshift=-0.7cm]180:(d)}] (D) {\includegraphics[scale=1]{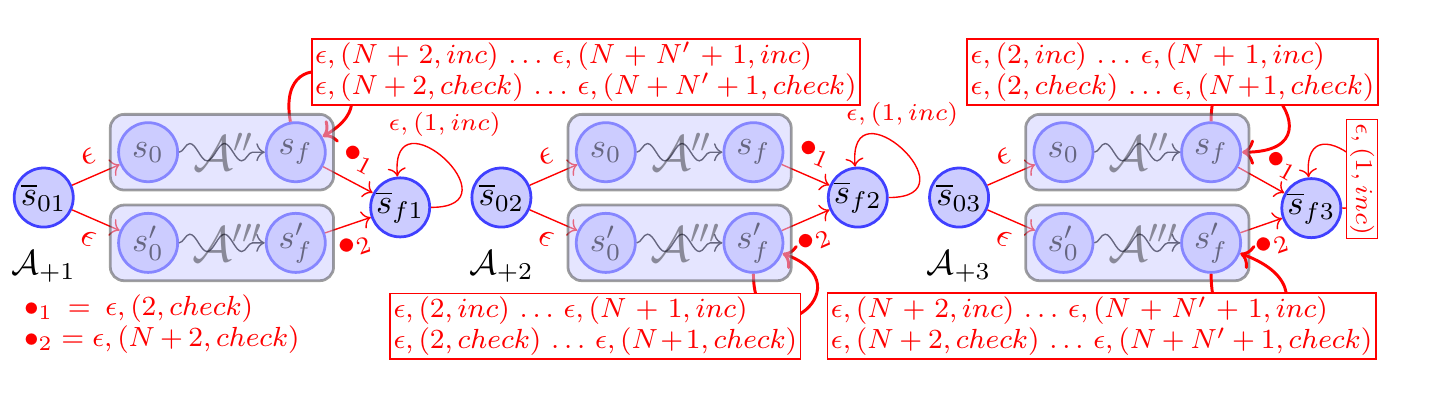}};
 \draw (D.west)node[above left=7mm and
    2mm,draw=gray!30,fill=gray!10,font=\footnotesize] {$\mathbf{e_1 + e_2}$};

\pgftransformshift{\pgfpoint{-3cm}{-2.7cm}}

\node[label={[yshift=-0.2cm]180:(e)}](E){
%\node[label={[label distance=-0.5cm]270:$t_1^*$}](A){
\includegraphics{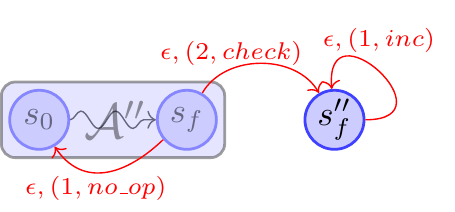}};
\draw (E)++(-2.85,.5)
node[draw=gray!30,fill=gray!10,font=\scriptsize]{$\mathbf{e_1^*}$};

\pgftransformshift{\pgfpoint{7cm}{-0cm}}

\node[label={[yshift=-0.2cm]180:(f)}](F){
%\node[xshift=4.25cm,label={[label distance=-0.5cm]270:$t_1^T$}](B){
\includegraphics{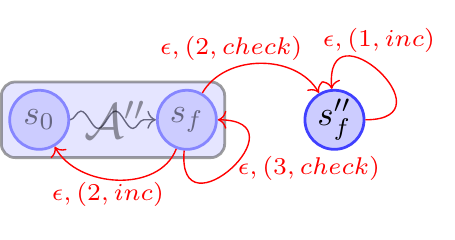}};
\draw (F)++(-2.85,.5)
node[draw=gray!30,fill=gray!10,font=\scriptsize]{$\mathbf{e_1^T}$};

\end{tikzpicture}
}%end scalebox
\vspace{-0.7cm}

\caption{\label{fig:translation}The automata for the translation of a $T$-regular expression $e$.}
\end{figure}

%\begin{figure}[!b]
%
%\centering
%\scalebox{.7}{%
%\begin{tikzpicture}
%\node(A){
%\includegraphics{emptytranslation}};
%\node[xshift=5cm](B){\includegraphics{symtranslation}};
%\node[xshift=3.2cm, yshift=-2cm](C){\includegraphics{concattranslation}};
%\node[xshift=13cm, yshift=2mm](C){\includegraphics{concattranslation}};
%\node[draw=gray!30,fill=gray!10,font=\footnotesize] at (-.75,-1.65) {$\mathbf{t_1 \cdot t_2}$};
%\draw (C)++(-3.95,.32) node[draw=gray!30,fill=gray!10,font=\footnotesize]{$\mathbf{e_1 \cdot e_2}$};
%
%\end{tikzpicture}
%}%end scalebox
%
%\caption{The automata $\cA_{\emptyset}$, $\cA_{a}$ for the base cases (left and middle pictures) and the automaton for the inductive step $e_1 \cdot e_2$ (right picture).}
%\label{fig:base}
%\end{figure}

%\begin{figure*}[t]
%
%\parbox[t][2cm]{\textwidth}{
%
%{\centering
%\includegraphics[scale=.96]{mixtranslation}
%
%}
%
%\vspace{-3.75cm}
%
%\hspace{7.5mm}
%\tikz\node[draw=gray!30,fill=gray!10,font=\footnotesize] {$\mathbf{t_1 + t_2}$};
%
%{\centering
%  \begin{tikzpicture}
%    \node (A) {\includegraphics[scale=.7]{mixtranslation}};
%
%    \draw (A.west)node[above left=7mm and
%    2mm,draw=gray!30,fill=gray!10,font=\footnotesize] {$\mathbf{e_1 + e_2}$};
%  \end{tikzpicture}
%
%}
%
%}
%
%\vspace{-5mm}
%
%\caption{The automata $\cA_{+_1}$, $\cA_{+_2}$, and $\cA_{+_3}$
%  (inductive step $e_1 +e_2$).}
%\label{fig:mix}
%\end{figure*}

\begin{compactitem}[\labelitemi\leftmargin=0.1em]
\item
  Let $e=e_1 \cdot e_2$, $\cA= (S,\Sigma,s_0,s_f,N,\Delta) \in \mathcal S_{e_1}$, and
$\cA'= (S', \Sigma,s_0', s'_f,N',\Delta') \in \mathcal S_{e_2}$. Moreover, let
$\cA''= (S,\Sigma,s_0,s_f,N+1,\Delta'')$ and $\cA'''= (S',\Sigma,s_0',s'_f,N'+N+1,\Delta''')$ be the
$1$-shifted version of $\cA$ and the $N+1$-shifted version of $\cA'$, respectively.
We define $\cA {\cdot} \cA'= (S\cup S'\cup \{s_f''\},\Sigma,s_0,s_f'',N+N'+1, \Delta''\cup \Delta'''
\cup \{(s_f,\epsilon, s'_0,(2,check)),(s_f',\epsilon, s''_f, (N+2,check)),(s_f'',\epsilon, s''_f,(1,inc)) \})$. \\
%
% where $s_f''$ is the final state of $\cA \cdot \cA'$. \\
%
We set $\mathcal S_{e_1 \cdot e_2} = \{ \cA {\cdot} \cA': \cA \in \mathcal
S_{e_1}, \cA'\in \mathcal S_{e_2}\}$.
See \figurename{}~\ref{fig:translation}~(c) for a graphical account.

\item
  Let $e=e_1 + e_2$,  $\cA= (S,\Sigma,s_0,s_f,N,\Delta) \in \mathcal S_{e_1}$, and $\cA'= (S',\Sigma,s_0',s'_f
N',\Delta') \in \mathcal S_{e_2}$. Moreover, let $\cA''$ and $\cA'''$ be defined as in the previous case.
%$\cA''= (S,\Sigma,s_0,N+1,\Delta'')$ and
%$\cA'''= (S',\Sigma,s_0',N'+N+1,\Delta''')$ be the
%$1$-shifted version of $\cA$ and the
%$N+1$-shifted version of $\cA'$, respectively.
%
%Let $A_1=(S'_1, F_1,\Sigma, s^1_0, \Delta'_1)$
%and  $A_2=(S'_2, F_2,\Sigma, s^2_0, \Delta'_2)$
%with $S_1\cap S_2=\emptyset$ the two NFA
%that recognize $\cL(t_1')$ and $\cL(t_2')$, respectively.
%For $i\in \{1,2\}$ we define the $CQ$ automata
%$\cA_i=(S_i = S'_i\cup\{ s^i_f\},\Sigma,s^i_0,1,
%\Delta_i = \{(s,\sigma, s', (1, no\_op)):
%(s,\sigma, s') \in \Delta'_i \}\cup \{(s,\sigma, s^i_f, (\epsilon, no\_op)):
%s\in F\})$.
%Given two automata $\cA= (S,\Sigma,s_0,C,\Delta) \in \mathcal S_{t_1}$ and
%$\cA'= (S',\Sigma,s_0',C',\Delta') \in \mathcal S_{t_2}$
%(we assume $S\cap S'=S\cap S_i= S'\cap S_i =\emptyset$ for $i \in \{1,2\}$),
We define $\cA{+}\cA'$ as the set $\{\cA_{+_1}, \cA_{+_2}, \cA_{+_3}\}$ (see \figurename~\ref{fig:translation}~(d)), where

\begin{compactitem}[\labelitemii\leftmargin=0.1em]
\item \sloppypar
  $\cA_{+_1}= (S\cup S'\cup\{ \os_{01}, \os_{f1}\}, \Sigma, \os_{01}, \os_{f1}, N'+N+1, \Delta''\cup
  \Delta'''
  \cup \{(\os_{01},\epsilon, s_0, (1, no\_op)), (\os_{01},\epsilon, s'_0,
(1,
\allowbreak
 no\_op)),
  (s_{f},\epsilon, \os_{f1}, (2,check)),(s'_{f},\epsilon,\os_{f1}, (N+2,
  check)),
  \allowbreak
  (\os_{f1},\epsilon, \os_{f1},(1, inc)) \}
  \cup \{ (s_f, \epsilon,
\allowbreak
 s_f, (k,*)): *\in \{inc, check\},N+2\leq k \leq
  N+N'+1\}),$
\item
  $\cA_{+_2}=(S\cup S'\cup\{ \os_{02}, \os_{f2}\}, \Sigma, \os_{02}, \os_{f2}, N'+N+1, \Delta''\cup
  \Delta'''
  \cup \{(\os_{02}, \epsilon, s_0,(1, no\_op)), (\os_{02},\epsilon, s'_0,
 (1,
\allowbreak
 no\_op)),
  (s_{f}, \epsilon, \os_{f2}, (2, check)),(s'_{f},\epsilon,\os_{f2}, (N +2,
  check)),
  \allowbreak
  (\os_{f2}, \epsilon, \os_{f2}, (1,inc)) \}
  \cup \{ (s'_f,\epsilon, s'_f, (k,*)): * \in \{inc, check\}, 2\leq k \leq
  N+1\})$, and
\item
  $\cA_{+_3}= (S\cup S'\cup\{ \os_{03}, \os_{f3}\}, \Sigma, \os_{03}, \os_{f3}, N'+N+1, \Delta''\cup
  \Delta'''
  \cup \{(\os_{03},\epsilon, s_0, (1, no\_op)), (\os_{03},\epsilon,
  s'_0,
 (1,
\allowbreak
 no\_op)),
  (s_{f},\epsilon, \os_{f3}, (2,check)), (s'_{f},\epsilon,\os_{f3},
  (N +2, check)),
  \allowbreak
  (\os_{f3},\epsilon, \os_{f3}, (1, inc)) \}
  \cup \{ (s_f,\epsilon,
\allowbreak
 s_f, (k,*)): *\in \{inc, check\}, 2\leq k \leq N+1 \}
  \cup \{ (s'_f, \epsilon, s'_f, (k,*)): *\in \{inc, check\}, N+2\leq k \leq
  N+N'+1 \} )$.
\end{compactitem}
%
%The final state of $\cA_{+_i}$ is $\os_{fi}$, for $1 \leq i \leq 3$. \\
%We define $\mathcal S_{t_1+t_2}=\!\!\!\!\!\! \bigcup\limits_{ \cA \in \mathcal S_{t_1}, \cA'\in \mathcal S_{t_2}}
%\!\!\!\!\!\!\!\!\cA +\cA'$.
%
We set $\mathcal S_{e_1+e_2} = \bigcup_{ \cA \in \mathcal S_{e_1}, \cA'\in \mathcal S_{e_2}} \cA {+}\cA'$.

\item
Let $e=e_1^*$, $\cA= (S,\Sigma,s_0,s_f,N,\Delta) \in \mathcal S_{e_1}$, and
$\cA''$ be defined as in the previous cases.
We let $\cA_* = (S\cup \{s_f''\}, \Sigma, s_0,s_f'',N+1, \Delta''\cup \{(s_f,
\epsilon, s_0, (1,
\allowbreak
no\_op)), (s_f, \epsilon, s_f'', (2,check)),
(s_f'',\epsilon, s_f'', (1,inc)) \})$. \\
%
%where $s_f''$ is the final state. \\
%
We set $\mathcal S_{e_1^*} = \{ \cA_*: \cA \in \mathcal S_{e_1}\}$.
See \figurename{}~\ref{fig:translation}~(e) for a graphical account.

\item
Let $e=e_1^T$ and $\cA= (S,\Sigma,s_0,s_f,N,\Delta) \in \mathcal S_{e_1}$. Moreover,
let $\cA'' =(S,\Sigma, s_0, s_f, N+2, \Delta'')$ be the $2$-shifted version of $\cA$. We let
$\cA_T = (S\cup\{s_f''\}, s_0, s_f'', N+2, \allowbreak \Delta''\cup
\{ (s_f,\epsilon, s_0, (2,inc)), (s_f,\epsilon, s_f'',(2,
\allowbreak
check)),
\allowbreak
(s_f,\epsilon,
\allowbreak
s_f, (3, check)), (s_f'', \epsilon, s_f'', (1, inc)) \})$. \\
%
%where $s_f''$ is the final state. \\
We set $\mathcal S_{e_1^T} = \{ \cA_T: \cA \in \mathcal S_{e_1} \}$.
See \figurename{}~\ref{fig:translation}~(f) for a graphical account.
\end{compactitem}

The next lemma states the correctness of the proposed encoding
%
% (proof in Appendix~\ref{app:proof-automata-encoding}).
%
(proof omitted for lack of space).
\newcommand{\lemmatranslation}{
  Let $e$ be a $T$-regular expression and $\cS_e$ be the corresponding set of
  automata. It holds:

  {\centering

    $\cL(e)= \bigcup_{\cA \in \cS_e} \cL_s (\widehat \cA)$.

  }
}
\newcounter{lemmatranslationcounter}
\setcounter{lemmatranslationcounter}{\thetheorem}
\begin{lemma}\label{lem:translation}
\lemmatranslation
\end{lemma}

%%(proof in Appendix~\ref{app:missing-proofs}).
%\newcommand{\lemtranslation}{
%%
%%Let $t$ be a $T$-regular expression and $\cS_t = \{ \cA_1,\ldots, \cA_n \}$ be
%%the corresponding set of automata, with $\cA_i= (S_i,\Sigma,s^i_0,s^i_f,N_i,\Delta_i)$ and $i \in \{ 1, \ldots, n \}$. It holds:
%%
%%\noindent
%%{\small
%%$\cL(t)= \hspace{-2.5mm} \bigcup\limits_{1 \leq i \leq n} \hspace{-2.5mm} \cL_s
%%((S_i, \Sigma,s^i_0,s^i_f,N_i,\Delta_i\cup\{ (s^i_f,\epsilon, s^i_0, (1,check))\}))$}.
%Let $t$ be a $T$-regular expression and $\cS_t$ be
%the corresponding set of automata. It holds:
%
%{\centering
%$\cL(t)= \bigcup_{\cA \in \cS_t} \cL_s
%(\widehat \cA)$.
%
%}
%%
%%$\cL(t)=
%%\bigcup_{\cA_i = (S_i,\Sigma,s^i_0,N_i,\Delta_i) \in \cS_t}
%%%\!\!\!\!\!\!\!\!\!\!\!\!\!\!\!\!\!\!
%%%\bigcup\limits_{\cA_i = (S_i,\Sigma,s^i_0,N_i,\Delta_i) \in \cS_t}
%%%\!\!\!\!\!\!\!\!\!\!\!\!\!\!\!\!\!\!
%%\cL_s ((S_i,\Sigma,s^i_0,N_i,\Delta_i\cup\{ (s^i_f,\epsilon, s^i_0, (1,check))\}))$.
%}
%
%\newcounter{lemtranslationcounter}
%\setcounter{lemtranslationcounter}{\thetheorem}
%\begin{lemma} \label{lem:translation}
%  \lemtranslation
%\end{lemma}
%
%
%\setcounter{theorem}{\thelemtranslationcounter}
%\begin{lemma}
%\lemtranslation
%\end{lemma}

%\subsection{Encoding of \texorpdfstring{$\omega T$}{omega T}-regular expressions}

%\paragraph{Encoding of \texorpdfstring{$\omega T$}{omega T}-regular
%  expressions.}

\noindent{\bf Encoding of \texorpdfstring{$\omega T$}{omega T}-regular
  expressions.}
We are now ready to deal with $\omega T$-regular expressions
%, generated by the grammar  $T \ ::= \ T + T \ \mid \ R.T \ \mid \ t^\omega$
(sub-grammar rooted in $E$ in paragraph ``$\omega T$-regular languages'' at page~\pageref{page:Tregexp-grammar}).
%As basis for our inductive construction we assume to have the set $S^t_{\cA}$
%for every sub-expression $t$ built in the construction above.
%
We must distinguish three cases.
\begin{compactitem}
\item 
If $E= E_1 + E_2$, then $\mathcal S_{E_1+E_2}$ is equal to $\mathcal S_{E_1}
\cup \mathcal S_{E_2}$.
\item
  If $E=R \cdot E'$, then let $A_R=(S_R, F_R,\Sigma, s^R_0, \Delta_R)$ be
the NFA that recognises the regular language $\cL(R)$, and $\cA= (S,\Sigma,s_0,
s_f, N,\Delta) \in \mathcal S_{E'}$.
We let $A_R {\cdot} \cA = (S\cup S_R, \Sigma, s^R_0, s_f, N,
\Delta \cup \{(s,\sigma, s', (1, no\_op)): (s,\sigma, s') \in \Delta_R \}
\cup \{(s,\epsilon, s_0, (1, no\_op)): s\in F_R\})$.
%, with final state $s_f$.
We set $\mathcal S_{R \cdot E'} = \{ A_R {\cdot} \cA: \cA \in \mathcal S_{E'}\}$.
%
%\item
%  Finally, let $T=t^{\omega}$. We define $\mathcal S_{t}$ as $\{\cA_1, \ldots, \cA_n\}$,
%where $\cA_i=(S_i,\Sigma,s^i_0,N_i,\Delta_i)$ and $s^i_f$ is the final state
%of $\cA_i$, for all $1\leq i\leq n$.
%$\mathcal S_{t^\omega}$ is the set $\{ (S_i,\Sigma,s^i_0,N_i, \Delta_i
%\cup\{ (s^i_f,\epsilon, s^i_0,(1, check) ) \} ): 1\leq i\leq n \}$.
\item
  Finally, if $E=e^{\omega}$, then
$\mathcal S_{e^\omega}$ is the set $\{ \widehat{\cA} : \cA \in \mathcal S_{e} \}$.
%\begin{compactitem}
%\item Let $T= T_1 + T_2$. $\mathcal S_{T_1+T_2}$ is equal to $\mathcal S_{T_1}
%\cup \mathcal S_{T_2}$.
%
%\item Let $T=R \cdot T'$. Moreover, let $A_R=(S_R, F_R,\Sigma, s^R_0, \Delta_R)$ be
%the NFA that recognises the regular language $\cL(R)$ and let $\cA= (S,\Sigma,s_0,
%N,\Delta) \in S_{T'}$.
%We define $A_R \cdot \cA$ as the automaton $(S\cup S_R, \Sigma, s^R_0, N,
%\Delta \cup \{(s,\sigma, s', (1, no\_op)): (s,\sigma, s') \in \Delta_R \}
%\cup \{(s,\epsilon, s_0, (1, no\_op)): s\in F_R\})$, with final state $s_f$.
%$\mathcal S_{R \cdot T'}$ is the set $\{ A_R \cdot \cA: \cA \in \mathcal S_{T'}\}$.
%
%\item Let $T=t^{\omega}$. We let $\mathcal S_{t}=\{\cA_1, \ldots, \cA_n\}$,
%where $\cA_i=(S_i,\Sigma,s^i_0,N_i,\Delta_i)$ and $s^i_f$ is the final state
%of $\cA_i$, for every $1\leq i\leq n$.
%$\mathcal S_{t^\omega}$ is the set $\{ (S_i,\Sigma,s^i_0,N_i, \Delta_i
%\cup\{ (s^i_f,\epsilon, s^i_0,(1, check) ) \} ): 1\leq i\leq n \}$.
%\end{compactitem}
%
%Therefore, for any given $\omega T$ regular expression $T$, it is possible
%to construct the corresponding set $\mathcal S_{T}$ of $CQ$ automata.
\end{compactitem}

As in the case of $T$-regular expressions, it is easy to check that,  for all $\omega T$-regular expressions $E$:

{\centering
  $\cL(E) = \bigcup_{\cA\in \mathcal S_{E}} \cL(\cA)$.

}

To complete the reduction, we only need to show how to merge the automata in
$\mathcal S_E$
%(for any given $\omega T$-regular expression $E$)
into a single one
%$CQ$ automata
$\mathcal A_E$ accepting the language $\cL(E)$.
%To this end, for a given $\omega T$-regular expression $T$,
Let $\mathcal S_{E} = \{\cA_{1}, \ldots, \cA_{n}\}$,
with $\cA_{i}=(S_{i},\Sigma,s^{i}_0,N_{i},\Delta_{i})$,
for $1\leq i\leq n$, and let $N_{\max} = \max\{N_i:1\leq i \leq n\}$.
For each $1\leq i\leq n$,
%we define the transition relation
let $ \overline{\Delta}_{i}= \Delta_{i} \cup \{ (s^i_0,\epsilon, s^i_0,(k, *)): *\in \{inc, check\}, N_i < k \leq N_{\max} \}$.
Finally, let $s_0$ be a fresh state. We define $\cA_{E}$ as the automaton
$(\bigcup_{1\leq i\leq n} S_{i}\cup \{s_0\}, \Sigma, s_0 ,N_{\max},
\bigcup_{1\leq i\leq n}(\overline{\Delta}_{i} \cup \{(s_0,\epsilon,
s^{i}_0,(1, no\_op) ) \}))$.

\begin{theorem}\label{thm:encoding}
For every $\omega{T}$-regular expression $E$, there exists a \cca ${\cal A}$
such that $\cL(E) = \cL({\cal A})$.
\end{theorem}
\section{From \texorpdfstring{$\omega{T}$}{omega T}-regular languages to \sonesU}
\label{sec:logic}

In this section, we provide an encoding of $\omega{T}$-regular expressions into \sonesU.

\noindent{\bf Definition of \sonesU.}
%The logic \sonesU is the extension of \sones with the unbounding quantifier
%\Uquantifier~\cite{bojan11tcs}.
The logic \sones is \mso interpreted over infinite words.
Its formulas are built over a finite, non-empty alphabet
$\Sigma$ and sets $V_1$ and $V_2$ of first- and second-order variables,
respectively:

{\centering
$
\begin{array}
  {@{\hspace{0mm}}l@{\hspace{1.9mm}}l@{\hspace{1.9mm}}l@{\hspace{0mm}}}
  \varphi & ::=
  & \tau \in P_\sigma \ \mid \ \tau \in X \ \mid \ \neg \varphi \ \mid \ \varphi
    \vee \varphi \ \mid \ \exists x. \varphi \ \mid \ \exists X. \varphi \\
  \tau & ::= & x \ \mid \ s(\tau)
\end{array}$

}

\noindent
where $\sigma \in \Sigma$, $x \in V_1$, and $X \in V_2$.
We denote by $V_\Sigma$ the set $\{ P_\sigma \mid \sigma \in \Sigma
\}$.\footnote{We also use the formulation $P(\tau)$ in stead of $\tau \in P$ ($P
  \in V_2 \cup V_\Sigma$).}
Technically, elements of $V_\Sigma$ are second-order variables (i.e., they range
over sets of positive natural numbers), but with a standard intended semantics:
they partition $\bbNp$ and an interpretation for them $\mathcal I : V_\Sigma
\rightarrow 2^{\bbNp}$ identifies an infinite word $w^{\mathcal I}$ over
$\Sigma$ as follows: $w^{\mathcal I}[i] = \sigma$ iff $i \in \mathcal
I(P_\sigma)$, for every $i \in \bbNp$, $\sigma \in \Sigma$.
Notice also that variables in $V_\Sigma$ always occur free (i.e., not bound by
any quantifier).
A formula is \emph{closed} if the only free variables are the ones in $V_\Sigma$;
otherwise, it is \emph{open}.
The semantics of a closed formula $\varphi$, denoted by $\semanticsvarphi$, is
the set of all infinite words that \emph{satisfy} $\varphi$, i.e,
$\semanticsvarphi = \{ w^{\mathcal I} \mid \mathcal I \models \varphi \}$.

The logic \sonesU extends \sones with the \emph{unbounding quantifier}
\Uquantifier, which is defined as in~\cite{bojan11tcs}:

{\centering

  $\Uquantifier X.\varphi(X) := \bigwedge_{n \in \bbN}
  \exists_{\mathsf{fin}} X (\varphi(X) \wedge |X| \geq n)$.

}

\noindent
where $\exists_{\mathsf{fin}}$ allows for existential quantification over finite
sets, i.e., $\exists_{\mathsf{fin}} X . \varphi \equiv \exists X . (\varphi \wedge \exists y
. X \subseteq \{ 1, \ldots, y \})$ for every second-order
variable $X$ and \sonesU-formula $\varphi$; the universal quantifier
$\forall_{\mathsf{fin}}$ is defined as the dual of $\exists_{\mathsf{fin}}$.
Intuitively, \Uquantifier makes it possible to say that a formula
$\varphi(X)$ (containing at least one second-order free variable $X$) is
satisfied by infinitely many finite sets and there is no bound on their sizes.
The \emph{bounding quantifier} \Bquantifier is defined as the negation of
\Uquantifier:
$\Bquantifier X. \varphi(X) := \neg \Uquantifier X. \varphi(X)
\equiv \bigvee_{n \in \bbN} \forall_{\mathsf{fin}} X (\varphi(X) \rightarrow |X|
< n)$.
%
%
%\smallskip
%
%{\centering
%  $
%  \begin{array}{l@{\hspace{1mm}}c@{\hspace{1mm}}l}
%    \Bquantifier X. \varphi(X)
%    & :=
%    & \neg \Uquantifier X. \varphi(X) \\
%    & \equiv
%    & \bigvee_{n \in \bbN} \forall_{\mathsf{fin}} X (\varphi(X) \rightarrow |X| < n).
%  \end{array}$
%
%}
%
%\smallskip
%
%\noindent
%
Its intended meaning is: there is a bound on the sizes of finite sets that
satisfy $\varphi(X)$.

\smallskip

\noindent{\bf Encoding.}
In what follows, given an $\omega T$-regular expression $E$ we show how to build
a formula $\varphi_E$ for which $\cL(E) = \semantics{\varphi_E}$.
For the lack of space, we only give an intuitive idea of the encoding.

For every $\omega T$-regular (sub-)expression $E$, let $E_{[T \mapsto *]}$ be the
$\omega$-regular (sub-)expression obtained from $E$ by replacing the $T$-constructor
with the $*$-constructor (e.g., if $E= (a^Tb)^\omega$, then $E_{[T \mapsto *]} =
(a^*b)^\omega$) and $\varphi_{E_{[T \mapsto *]}}$ be the \sones-formula for
which $\semantics{\varphi_{E_{[T \mapsto *]}}} = \cL({E_{[T \mapsto *]}})$ holds
(its existence is guaranteed by the equivalence between \sones 
and $\omega$-regular languages).

Let $E$ be an $\omega T$-regular expression.
In order to correctly define $\varphi_{E}$ we need to enrich such a formula
$\varphi_{E_{[T \mapsto *]}}$ to enforce the condition imposed by occurrences of the
$T$-constructor in $E$.
%, for every sub-expression of $t$ of the form $e^T$.
%
%More precisely, we need to
%
The intuitive idea is to control, for every sub-expression $e^T$, the sizes of
\emph{$e$-blocks} (i.e., maximal blocks of consecutive occurrences of finite
words in $\cL(e_{[T \mapsto *]})$) along infinite words.
(Notice that $e_{[T \mapsto *]}$ is a regular expression.)
According to the semantics of the $T$-constructor, we have to force the
existence of $e$-blocks of infinitely many different sizes, and infinitely many
of such sizes must occur infinitely often.
To this end, given a regular expression $e$, we build a formula \Tcondition{e}
that is satisfied by an infinite word $w$ iff there are infinitely many $k \in
\bbN$ such that $w$ features infinitely many $e$-blocks of size $k$.
In our construction, we use formulas $\isregexp{e}(x,y)$ (for every regular
expression $e$), featuring two free first-order variables, with the
following semantics: $w$ satisfies $\isregexp{e}[x \mapsto \bar x,y \mapsto
\bar y]$ iff $w[\bar x, \bar y] \in \cL(e)$.
In addition, we use the unary predicate $\isbeginningofexp{e}(x)$, with the
following semantics: $w$ satisfies $\isbeginningofexp{e}[x \mapsto \bar x]$ iff
$w[\bar x, \bar y] \in \cL(e)$ for some $\bar y \in \bbNp$.

Let $e$ be a regular expression.
To begin with, we define formula $\block{e}(X)$, stating that $X$ is a maximal
set of positions from which consecutive sub-words belonging to $\mathcal L(e)$
begin; roughly speaking, $X$ is an $e$-block.

\smallskip

\noindent
{\centering
$\begin{array}{@{\hspace{0mm}}l@{\hspace{.1mm}}l@{\hspace{.1mm}}l@{\hspace{0mm}}}
\block{e}(X)
& :=
& \exists y \exists z . [
 \isregexp{e^*}(y,z)
\wedge X \subseteq \{ y, \ldots, z \}
\wedge \forall x . (x \in \{ y,\ldots, z \} \wedge \isbeginningofexp{e}(x)
\rightarrow x \in X ) ].
\end{array}$

}

%\noindent
%$\consecutive{e}(X) \equiv (\forall x \in X . \isbeginningEx) \wedge
%(\forall x < y < z . x,z \in X \rightarrow \isinexprEy ) $

\smallskip

\noindent Next formula $\blockset{e}(Y)$ says that
\begin{inparaenum}[$(i)$]
\item $Y$ only contains $e$-blocks,
\item it contains infinitely many of them, and
\item there is an upper bound on their sizes.
\end{inparaenum}
In this case, we say that $Y$ is an \emph{$e$-block-set}.

\smallskip

\noindent
{\centering
$\begin{array}{@{\hspace{0mm}}l@{\hspace{0mm}}}
   \blockset{e}(Y) :=
   [\forall y . (y \in Y \rightarrow \exists_{\mathsf{fin}} X . (\block{e}(X)
   \wedge X \subseteq Y \wedge y \in X))] \wedge \\
   \hspace{24mm}
   \wedge [\forall y \exists_{\mathsf{fin}} X . ( \block{e}(X) \wedge X \subseteq Y
   \wedge \min X > y )] \wedge [\Bquantifier x . (X \subseteq Y \wedge \block{e}(X))] \\
\end{array}$

}

\smallskip

\noindent Finally, we define \Tcondition{e} as
   $\forall Y . [\blockset{e}(Y) \rightarrow
   \exists Z . ( \blockset{e}(Z) \wedge Y \subsetneq Z \wedge \exists^\omega x
   . x \in Z \setminus Y) ]$,
%
%
%\smallskip
%
%\noindent
%{\centering
%$\begin{array}{@{\hspace{0mm}}l@{\hspace{0mm}}}
%   \forall Y . [\blockset{e}(Y) \rightarrow \\
%   \hspace{4mm}
%   \exists Z . ( \blockset{e}(Z) \wedge Y \subsetneq Z \wedge \exists^\omega x
%   . x \in Z \setminus Y) ].
%\end{array}$
%
%}
%
%\smallskip
%
where $\exists^\omega$ allows for infinite existential first-order
quantification, i.e., $\exists^{\omega} x . \varphi \equiv \forall y \exists x
. (x > y \wedge \varphi)$.

\newcommand{\lemmalogic}{
  Let $e$ be a regular expression.
  An infinite word $w$ satisfies \Tcondition{e} iff there are infinitely many
  natural numbers $k$ such that $w$ features infinitely many $e$-blocks of size
  $k$.
}
\newcounter{lemmalogiccounter}
\setcounter{lemmalogiccounter}{\thetheorem}
\begin{lemma}
%[proof in Appendix~\ref{app:proof-logic-encoding}]
\label{lem:logic}
\lemmalogic
\end{lemma}
\begin{proof}
  Let $w$ be an infinite word that satisfies \Tcondition{e} and let us assume
  that there are only finitely many natural numbers $k$ such that infinitely
  many $e$-blocks of size $k$ occur in $w$.
  Let $k_{\max}$ be the largest among such numbers and let $\bar Y$ be the
  $e$-block-set containing all $e$-blocks of size not larger than $k_{\max}$.
  Clearly, $w$ satisfies $\blockset{e}[Y \mapsto \bar Y]$ and thus, by the
  definition of \Tcondition{e}, there exists an $e$-block-set $\bar Z$ that
  contains infinitely many $e$-blocks (of bounded size) that do not belong to
  $\bar Y$.
  Since $\bar Z \supsetneq \bar Y$ (that means $Z$ contains all $e$-blocks in
  $\bar Y$ as well), there exists a number $k' > k_{\max}$ such that infinitely
  many $e$-blocks of size $k'$ occur in $w$.
  This is in contradiction with our initial hypothesis that $k_{\max}$ is the
  largest number such that infinitely many $e$-blocks of size $k_{\max}$ occur in
  $w$, hence the thesis follows.

  In order to prove the converse direction, let us assume that there are
  infinitely many natural numbers $k$ such that infinitely many $e$-blocks of
  size $k$ occur in $w$ and let $\bar Y$ be an $e$-block-set (i.e., $w$ satifies
  $\blockset{e}[Y \mapsto \bar Y]$).
  By the definition of $\blockset{e}(Y)$ (in particular, the third conjunct),
  there is a bound on the size of all $e$-blocks in $\bar Y$.
  Let $k_{\max}$ be such a bound.
  By our assumption, there is a number $k' > k_{\max}$ such that infinitely many
  $e$-blocks of size $k'$ occur in $w$.
  Let $\bar Z$ be the set containing all $e$-blocks in $\bar Y$ and, in addition
  all $e$-blocks of size $k'$ occurring in $w$.
  Clearly, $\bar Z$ is an $e$-block-set that contains $\bar Y$ and feaures
  infinitely many elements not belonging to $\bar Y$ (i.e., $w$ satisfies the
  formula $(\blockset{e}(Z) \wedge Y \subsetneq Z \wedge \exists^\omega x . x \in
  Z \setminus Y)[Y \mapsto \bar Y, Z \mapsto \bar Z]$), and thus
  $w$ satisfies \Tcondition{e}.
\end{proof}

Making use of formulas \Tcondition{e}, for every regular expression $e$, it
is possible to strengthen $\varphi_{E_{[T \mapsto *]}}$ to enforce the
condition, imposed by occurrences of the $T$-constructor in $E$, on sizes of
$e$-blocks occurring in infinite words, for every sub-expression $e^T$ of $E$.
Thus, we can conclude the main result of this section.

\begin{theorem}
  For every $\omega T$-regular expression $E$, we have that
  $\semantics{\varphi_E} = \cL(E)$.
\end{theorem}

As a conclusive remark, notice that \Tcondition{e} uses quantification over
infinite sets, implying that $\varphi_t$ does not belong to the language of
\wsonesU, where second-order quantification is only allowed over finite sets.

\cut{
\hrule

We proceed as follows.

\begin{compactenum}
\item We define formula $\ismaxblockEX$ whose only free variable is $X$, such
  that $\semantics{\ismaxblockEX} = \{ (w,X) \mid w \in (E^*F)^\omega$ and $X$
  identifies a maximal block of repetitions of $E$ in a $\omega$-iteration in
  $w$

  $$\ismaxblockEX \equiv \forall x \in X \exists y > x. isE(x,y)
  \forall x,y,z (x,z \in X \rightarrow y \in X) $$

\end{compactenum}

$isE(x,y) \equiv x=y \wedge a(x)$ if $E = a$

$isE(x,y) \equiv \exists z. x \leq z \leq y \wedge isE_1(x,z) \wedge
isE_2(z+1,y)$ if $E = E_1E_2$

$isE(x,y) \equiv \exists z. x \leq z \leq y \wedge isE_1(x,z) \wedge
isE_2(z+1,y)$ if $E = E_1E_2$
}

%%% Local Variables:
%%% mode: latex
%%% TeX-master: "main"
%%% End:

%\input{strongT}
\section{Conclusions}\label{sec:conclusions}

In this paper, we introduced a new class of extended $\omega$-regular languages ($\omega{T}$-regular languages), 
that captures meaningful languages not belonging to the class of $\omega{BS}$-regular ones. We first gave a characterization of them in terms of $\omega{T}$-regular expressions. Then, we defined the new class of counter-check automata
(CCA), with a decidable emptiness problem, and we proved that they are expressive enough to capture them. Finally, we provided an embedding of \texorpdfstring{$\omega{T}$}{omega T}-regular languages in \sonesU.

%\smallskip

In the exploration of the space of possible extensions of $\omega$-regular languages, we studied
%took into consideration 
also a stronger variant of $(.)^T$, that forces $\omega$-words to feature infinitely many exponents, \emph{all of them} occurring infinitely often (a detailed account can be found in \cite{preprint201701}). To a large extent, the results obtained for $(.)^T$ can be replicated for this stronger variant. In particular, it is possible to introduce a new class of automata, called \emph{counter-queue automata} (\cqa),
%for short), 
that generalize \cca, whose emptiness problem can be proved to be decidable in 2ETIME and which are expressive enough to capture $\omega$-regular languages extended with the stronger variant of $(.)^T$.  As in the case of $\omega{T}$-regular languages, the problem of establishing whether or not  the new languages are expressively complete with respect to \cqa is open.
There are, however, at least two significant differences between$(.)^T$ and its stronger variant. First, $(.)^T$ satisfies the following property of \emph{prefix independence}.  Let $e$ be a $T$-regular expression and let $\vec{u} = (u_1, u_2, \ldots)$ and $\vec{v} = (u_h, u_{h+1}, \ldots)$ be two word sequences such that $\vec{v}$ is the infinite suffix of  $\vec{u}$ starting at position $h$ and $u_i \in \mathcal{L}(e)$ for all $i$. Then, $\vec{u} \in \mathcal{L}(e^T)$ iff $\vec{v} \in \mathcal{L}(e^T)$. Both $(.)^B$ and $(.)^S$ satisfy an analogous property, while this is not the case with the stronger variant of $(.)^T$: if $\vec{u}$ belongs to the language, then $\vec{v}$ belongs to it as well, but not vice versa. The second difference is that there seems to be no way to generalize the embedding of \texorpdfstring{$\omega{T}$}{omega T}-regular languages into \sonesU given in Section \ref{sec:logic} to the stronger variant of $(.)^T$. 
%(as a matter of fact, we strongly believe  such a variant  not to be definable in \sonesU).

%The class of $\omega{T}$-regular languages can be viewed as an intermediate step towards the study of the class of 
%$\omega{BST}$-regular languages, obtained from the combination of  $\omega{T}$- and $\omega{BS}$-regular ones. In 
%particular, we would like to establish whether $\omega{BST}$-regular languages are closed under complementation.

%\smallskip

As for future work, we would like to investigate different combinations of $(.)^B$, $(.)^S$, and (weak and strong) $(.)^T$
We already know that $\omega{BST}$-regular languages are not closed under complementation. Indeed, if they were, they would be expressively complete for \sonesU. However, it is known from~\cite{DBLP:journals/fuin/HummelS12} that \sonesU makes it possible to define languages that are complete for arbitrary levels of the projective hierarchy, while $\omega{BST}$-regular languages live at the first level (analytic sets), and thus they cannot define full \sonesU.
A particularly interesting issue is the one about the intersections of $\omega{B}$-, $\omega{S}$-, and weak/strong $\omega{B}$-regular languages.
In \cite{DBLP:journals/corr/Skrzypczak14}, it has been shown that a language which is both $\omega{B}$- and $\omega{S}$-regular is also
$\omega$-regular. We aim at providing a characterization of languages which are both $\omega{B}$- 
(resp.,  $\omega{S}$-) and $\omega{T}$-regular.
%
%\smallskip
%
We are also interested in (modal) temporal logic counterparts of extended $\omega$-regular languages. To the best of our knowledge, none was provided in the literature. We started to fill such a gap in \cite{lics2013,DBLP:conf/lata/MontanariS13}.

%%% Local Variables:
%%% mode: latex
%%% TeX-master: "main"
%%% End:

\paragraph{Acknowledgements.} We would like to acknowledge prof.~Massimo
Benerecetti for some helpful comments on the relationship between *-, B-, S-,
and T-constructors.

%\nocite{*}
\bibliographystyle{eptcs}
\bibliography{biblio}

%% .. or use the thebibliography environment explicitely

%\newpage
%\appendix
%\input{appendixproofPropositionKleene}
%\input{appendixproofWinning}
%\input{appendixexample}
%\input{appendixproofAutomataEncoding}
%\input{appendixproofLogicEncoding}

%\input{appendix3_1_old_T_definition}
%\input{appendix3_2_old_automata}
%\input{appendix3_3_old_nonempty_mpda}

\end{document}